\documentclass[11pt]{article}
\usepackage[margin=1in]{geometry}
\usepackage{authblk}
\usepackage{xcolor}
\usepackage{libertine}
\definecolor{ForestGreen}{rgb}{0.1333,0.5451,0.1333}
\definecolor{DarkRed}{rgb}{0.65,0,0}
\definecolor{Red}{rgb}{1,0,0}
\usepackage{backref}
\usepackage[linktocpage=true,backref=page, pagebackref=true,colorlinks, linkcolor=DarkRed,citecolor=ForestGreen, bookmarks,bookmarksopen,bookmarksnumbered]{hyperref}

\makeatletter
\renewcommand\AB@affilsepx{, \protect\Affilfont}
\makeatother

\title{The Average-Value Allocation Problem\thanks{David Wajc is partially supported by a Taub Family Foundation ``Leader in Science and Technology'' fellowship and ISF grant 3200/24. Anupam Gupta is supported in part by NSF awards CCF-1955785 and CCF-2006953. Part of this work was done while David Wajc was visiting Google Research, and Anupam Gupta was with Carnegie Mellon University.}}
\author[1]{Kshipra Bhawalkar}
\author[1]{Zhe Feng}
\author[1,2]{Anupam Gupta}
\author[1]{Aranyak Mehta}
\author[3]{David Wajc}
\author[1]{Di Wang}
\affil[1]{Google Research}
\affil[2]{NYU}
\affil[3]{Technion}
\date{\vspace{-1.5cm}}

\usepackage[utf8]{inputenc}

\usepackage{todonotes}
\usepackage{comment}

\usepackage{latexsym,graphicx}
\usepackage{amsmath,amsthm,amssymb,enumitem}
\usepackage{empheq}
\usepackage{xspace}
\usepackage{bm}
\usepackage{ifpdf}
\usepackage{color}
\usepackage{algorithm}
\usepackage{algorithmicx}
\usepackage[noend]{algpseudocode}
\usepackage{nicefrac}
\usepackage{wrapfig}
\usepackage{dsfont}
\usepackage{thm-restate}

\allowdisplaybreaks

\usepackage{cleveref}

\newtheorem{theorem}{Theorem}[section]
\newtheorem{definition}[theorem]{Definition}

\newtheorem{lemma}[theorem]{Lemma}
\newtheorem{fact}[theorem]{Fact}

\newtheorem{example}[theorem]{Example}

\newtheorem{remark}[theorem]{Remark}
\numberwithin{algorithm}{section}

\usepackage[all]{hypcap}

\usepackage{cleveref}

\newcommand{\junk}[1]{}
\newcommand{\ignore}[1]{}

\newcommand{\E}[0]{{\ensuremath{\mathbb{E}}}}

\newcommand{\poly}{\operatorname{poly}}

\newcommand{\sse}{\subseteq}

\newcommand{\calA}{{\mathcal{A}}}

\newcommand{\calE}{{\mathcal{E}}}
\newcommand{\calI}{{\mathcal{I}}}
\newcommand{\calM}{{\mathcal{M}}}

\newcommand{\eps}{\varepsilon}

\newcommand{\pr}[1]{{\mathds{P}r} \left[ #1 \right]}
\newcommand{\ex}[1]{{\mathbb{E}} \left[ #1 \right]}
\newcommand{\Ber}{\textrm{Bernoulli}}
\newcommand{\Bin}{\textrm{Binomial}}

\renewcommand{\Pr}{\mathds{P}r}

\newcounter{note}[section]

\newcommand{\alert}[1]{{\color{red}#1}}

\newcommand{\qedsymb}{\hfill{\rule{2mm}{2mm}}}

\newcommand{\initOneLiners}{%
    \setlength{\itemsep}{0pt}
    \setlength{\parsep }{0pt}
    \setlength{\topsep }{0pt}
}

\newcommand{\squishlist}{
 \begin{list}{$\bullet$}
  { \setlength{\itemsep}{0pt}
     \setlength{\parsep}{3pt}
     \setlength{\topsep}{3pt}
     \setlength{\partopsep}{0pt}
     \setlength{\leftmargin}{1.5em}
     \setlength{\labelwidth}{1em}
     \setlength{\labelsep}{0.5em} } }

\newcommand{\squishend}{
  \end{list}  }

\DeclarePairedDelimiterX{\infdivx}[2]{(}{)}{%
  #1\;\delimsize\|\;#2%
}

\newcommand{\nf}{\nicefrac}

\newcommand{\AVA}{AVA\xspace}
\newcommand{\GAVA}{GenAVA\xspace}

\newcommand{\pedge}{$P$-edge\xspace}
\newcommand{\nedge}{$N$-edge\xspace}
\newcommand{\pedges}{$P$-edges\xspace}
\newcommand{\nedges}{$N$-edges\xspace}
\newcommand{\pitem}{$P$-item\xspace}
\newcommand{\nitem}{$N$-item\xspace}
\newcommand{\pitems}{$P$-items\xspace}
\newcommand{\nitems}{$N$-items\xspace}

\begin{document}
\maketitle

\begin{abstract}
We initiate the study of centralized algorithms for welfare-maximizing allocation of goods to buyers subject to \emph{average-value constraints}. 
We show that this problem is NP-hard to approximate beyond a factor of $\frac{e}{e-1}$, and provide a $\frac{4e}{e-1}$-approximate offline algorithm. 
For the online setting, we show that no non-trivial approximations are achievable under adversarial arrivals.
Under i.i.d.~arrivals, we 
present a polytime online algorithm that provides a constant approximation of the optimal (computationally-unbounded) online algorithm. In contrast, we show that no constant approximation of the ex-post optimum is achievable by an online algorithm.
\end{abstract}

\newpage

\section{Introduction}
\label{sec:introduction}

Allocating goods to buyers to maximize social
welfare is one of the most central 
problems in economics. This problem, even under linear utilities, 
is complicated by buyers' various constraints and the manner 
in which items are revealed.

In this work we introduce the \emph{average-value allocation} problem
(\AVA).  Here, we wish to maximize social welfare (total value of
allocated items), while guaranteeing for each buyer $j$ an
\emph{average} value of allocated items of at least $\rho_j$.
Formally, if the value of item $i$ for buyer $j$ is $v_{ij}$, and
$x_{ij} \in \{0,1\}$ indicates whether item $i$ is allocated to buyer
$j$, we wish to maximize the social welfare,
$\sum_{ij} v_{ij}\;x_{ij}$, subject to each item being allocated to at most one buyer (i.e., $\sum_j x_{ij} \leq 1$), and to the ``average value'' constraint:
\begin{gather}
 \forall j,\;\;\;\; \sum_{i} v_{ij} \; x_{ij}
   \geq \rho_j \cdot \bigg(\sum_i x_{ij}\bigg). \label{eq:average}
\end{gather}

Average-value constraints arise naturally in numerous situations.
E.g., consider settings when goods are to be distributed among
``buyers'', and the (fixed) cost of distributing, receiving, or
deploying each such good allocated is borne by the recipient. Each
buyer wants their average value for their goods to be at least some
parameter $\rho_j$. This parameter $\rho_j$ allows to convert between
units, and so this fixed cost for each buyer can be in money, time,
labor, or any other unit.  So, for example, for allocation and
distribution of donations to a charitable organization, a certain
value-per-item is required to justify the time contributed by
volunteers, or the money spent by government in the form of subsidies.
In other words, the amount of ``benefit'' per task allocated to an
individual $j$ should be above the threshold $\rho_j$, so that even if
some of the tasks are individually less rewarding (i.e., they have
benefit less than $\rho_j$), the total amount of happiness they get
overall justifies their workload.

In addition to this average-value constraint on the allocation, we may
also consider side-constraints (such as the well-studied budget
constraints), but for now we defer their discussion and focus
on the novel Constraint \eqref{eq:average}. At first glance, the AVA
problem may seem similar to other packing problems in the literature,
but there is a salient difference---it is not a packing problem at
all! Indeed, if buyer $i$ gets some subset
$S_i = \{ j \mid x_{ij} = 1\}$ of items in some feasible allocation,
it is possible that a subset $S' \sse S_i$ of this allocation is no
longer feasible, since its average value may be lower. Given that this
packing (subset-closedness) property is crucial to many previous
results on allocation problems, their techniques do not apply. Hence, we have to examine this problem afresh, and we
ask: \emph{how well can the average-value allocation be approximated?}
We investigate this question, both in the offline and online settings.

\subsection{Our Results and Techniques}
\label{sec:our-results}

Recall that the \AVA problem seeks to maximize the social welfare $\sum_{ij} v_{ij} x_{ij}$ subject to each item going to at most one buyer, and also the novel average-value Constraint~\eqref{eq:average} above.
Our first result rules out polynomial-time exact algorithms for \AVA
in an offline setting, or even a PTAS, showing that this problem is as
hard to approximate as the \textsc{Max-Coverage} problem.
\begin{theorem}[Hardness of \AVA]
  For any constant $\eps > 0$, the \AVA problem is NP-hard to
  $(\frac{e}{e-1}-\eps)$-approximate. 
\end{theorem}

We then turn our attention to positive results, and give the following
positive result for the problem.
\begin{theorem}[Offline \AVA]
  \label{thm:offline}
  There exists a randomized polynomial-time algorithm for the \AVA
  problem which achieves an approximation factor of $\frac{4e}{e-1}$. 
\end{theorem}

To prove \Cref{thm:offline}, we would like to draw on techniques used
for traditional packing problems, but the non-traditional nature of
this problem means we need to investigate its structure carefully. A
key property we prove and leverage throughout is the existence of
approximately-optimal solutions of a very special kind: each buyer
gets a collection of ``bundles'', where a bundle for buyer $j$
consists of a single item $i$ with positive $v_{ij}-\rho_j$ (i.e.,
contributing positively to the average-value
Constraint~\eqref{eq:average}) and some number of items $i$ with
negative $v_{ij}-\rho_j$, such that they together satisfy the AVA
constraint. Given this structure we can focus on partitioning items
among bundles, and allocating bundles to buyers. Note that this
partitioning and allocation have to happen simultaneously, since the
values (i.e., $v_{ij}$) and whether it contributes positively or
negatively (i.e., $v_{ij}-\rho_j$) depend on the buyer and bundle
under consideration. We show how algorithms for GAP (generalized
assignment problem) with matroid
constraints~\cite{calinescu2011maximizing} can be used  to effectively perform both partitioning and allocation.

\paragraph{Relax-and-Round.} In order to extend our results from the
offline to the online settings, and to add in side-constraints, we
then consider linear programming (LP) based relax-and-round algorithms
for the \AVA problem. The LP relaxations take advantage of the
structural properties above, as they try to capture the best
bundling-based algorithms (and hence to approximate the optimal
solution of any kind). Once we have fractional solutions to the LP, we
can then round these in both offline and online settings to get our
feasible allocations.

Our first rounding-based algorithm, given in
\S\ref{sec:rounding-offline}, is in the offline setting, and yields
another $O(1)$-approximate algorithm for \AVA, qualitatively matching
the result from \Cref{thm:offline}. While the constants are weaker,
the result illustrates our ideas, and allows us to support additional
side-constraints (more on this in \S\ref{sec:generalizations}).

\paragraph{Online Algorithms.} We then turn to online \AVA, where
items arrive over $T$ timesteps, and must be allocated to buyers as
soon as they arrive. We want to maintain feasible solutions to the
\AVA at all times. We show that under adversarial arrivals, only
trivial $O(T)$ approximations are possible. We therefore shift our attention to i.i.d.~arrivals. Our first result is a time-efficient
approximation of the optimum (computationally-unbounded) online
algorithm:
\begin{theorem}[Online \AVA: Approximating the Optimal Online IID Algorithm]
  There exists a randomized polynomial-time online algorithm for the
  \AVA problem which achieves a constant factor of the value achieved
  by the optimal (computationally-unbounded) online algorithm.
\end{theorem}

To approximate the optimum online algorithm, we provide an LP
capturing a constraint only applicable to online algorithms, inspired
by such constraints from the secretary problem and prophet inequality
literatures \cite{buchbinder2014secretary,papadimitriou2021online}.
We then provide a two-phase online algorithm achieving a constant
approximation of this LP, analyzed via a coupling with an imaginary
algorithm that may violate \AVA constraints and allocate items to
several buyers. 

We then turn our attention to approximating the ex-post optimum
(a.k.a., getting a competitive ratio for the observed sequence).
In contrast, we show that when comparing with the ex-post optimum, no
such constant approximation ratio is possible, but we give matching
upper and lower bounds. (Due to lack of space, this is deferred to \Cref{sec:OPToff}.)
\begin{theorem}[Online \AVA: Ex-post Guarantees (Informal)]
  There exist families of online i.i.d.\ \AVA instances with $T$ arrivals on which any online algorithm is
  $\Omega\big(\frac{\ln T}{\ln \ln T}\big)$-competitive.
  In contrast, there exists an online algorithm matching this bound asymptotically (on all instances).
\end{theorem}
The lower bound is proved by giving an example using a balls-and-bins
process (and its anti-concentration). Then we formulate an LP
capturing this kind of anti-concentration, using which we match the
lower bound, under some mild technical conditions (see
\Cref{sec:OPToff} for details).

\subsubsection{Generalizations}
\label{sec:generalizations}

There are many interesting generalizations of the basic problem.  For
example, there might exist ``budgets'' which limit the number of items
any buyer can receive; or more generally we may have costs on items
which must sum to at most the buyer's budgets. These costs could be
different for different buyers, and in different units than those
captured by Constraint \eqref{eq:average}. These constraints are the
natural ones considered in packing problems; in general, we can
consider the \AVA constraint as being a non-packing constraint on the
allocation that can supplemented with other conventional packing
constraints.
As we show in \S\ref{sec:side-constraints}, our relax-and-round
algorithm extends seamlessly to accommodate such side constraints,
provided any individual item has small cost compared to the relevant
budgets.

Another natural generalization is \emph{return-on-spend} (RoS)
constraints, which have been central to much recent work on
advertisement allocation (see
\cite{googleautobiddingsupport,fbautobiddingsupport}) and
\S\ref{sec:related-work}).  We call the problem \emph{generalized
  \AVA} (\GAVA) and define it as follows: the objective is to maximize social
welfare, but now the average value is measured in a more general
way. Indeed, the allocation of item $i$ to buyer $j$ can incur a
different ``cost'' $c_{ij}$, and the average-value constraint becomes
the following ROS constraint:
\begin{gather}
 \forall j,\;\;\;\; \sum_{i} v_{ij} \; x_{ij} \geq \rho_j \cdot \bigg(
  \sum_i \alert{\mathbf{c_{ij}}}\; x_{ij}\bigg). \label{eq:ros}
\end{gather}
In contrast to \AVA, we show that allowing general costs $c_{ij}$ in
the generalized \AVA problem in~(\ref{eq:ros}) makes it as hard as one
of the hardest combinatorial problems---computing a maximum clique in
a graph. In particular, we show that it is NP-hard to
$n^{1-\eps}$-approximate \GAVA with $n$ buyers, for any constant
$\eps>0$.  In \Cref{sec:hardness} we show that similar hardness
persists even for stochastically generated inputs, and the problem
remains hard even if we allow for bicriteria approximation.

\subsection{Related Work}\label{sec:related-work}

Resource allocation is one of the most widely-studied topics in theoretical computer science. Here we briefly discuss some relevant lines of work.

\vspace{-0.2cm}
\paragraph{Packing/Covering Allocation Problems.}
The \emph{budgeted allocation problem} or \textsc{AdWords} of
\cite{mehta2007adwords} is NP-hard to approximate
within some constant \cite{chakrabarty2010approximability}, and constant
approximations are known even online~\cite{mehta2007adwords,buchbinder2007online,huang2020adwords}. The \emph{generalized assignment problem (GAP)}
\cite{feldman} and its extension, the \emph{separable assignment problem}, have constant approximations in both
offline~\cite{fleischer2011tight,calinescu2011maximizing} and (stochatic) online
settings \cite{KRTV18}. In both cases, arbitrarily-good approximations are 
impossible under adversarial online arrivals, even
under structural assumptions allowing for
an offline PTAS (e.g., ``small'' bids) \cite{mehta2007adwords}. However, assuming both small bids
and \emph{random-order} (or \emph{i.i.d.}) arrivals allows us to achieve
$(1-\eps)$-competitiveness~\cite{DevanurHayes09, Devanur11, KRTV18,
  GM16, agrawal2014fast}. Some such allocation problems are also
considered with concave or convex utilities~\cite{DevanurJ12,azar2016online}.
As noted above, many results and techniques for (offline and online)
packing and covering constraints are not applicable to our problem,
which is neither a packing nor covering problem in the conventional sense.

\vspace{-0.2cm}
\paragraph{RoS constraints in online advertising.} Return-on-spend
constraints as defined in~(\ref{eq:ros}) have received much attention
in recent years in the context of online advertising.  Several popular
autobidding products allow advertisers to provide campaign-level RoS
constraints with a goal to maximize their volume or value of
conversions (sales)
\cite{googleautobiddingsupport,fbautobiddingsupport}). 
Fittingly, there has been much
interest in understanding the RoS setting along various directions,
including optimal bidding \cite{aggarwal2019autobidding},
mechanism design  \cite{BalseiroDMMZ21Neurips,GolrezaiLP21},
and on welfare properties at equilibrium
\cite{aggarwal2019autobidding,DengMMZ21,Mehta22}.
In these results, distributed bidding based algorithms are shown to
achieve a constant fraction of the optimal welfare.  
However, note that the per-item costs in the autobidding setting are \emph{endogenous} (set via auction dynamics) whereas in our allocation problem there is no pricing mechanism and the costs are \emph{exogenous}. 
Our results about the hardness of the generalized \AVA show that under exogenous prices, such allocation problems do not admit constant (or even sublinear) approximation guarantees.

\vspace{-0.2cm}
\paragraph{Approximating the optimum online algorithm.} 
Our online i.i.d.~results relate to a recent burgeoning line of work on  approximation of the optimum online algorithm in stochastic settings via restricted online algorithms.
This includes restriction to polynomial-time algorithms (as in our case) in discrete settings 
\cite{niazadeh2018prophet,anari2019nearly,papadimitriou2021online,braverman2022max,dutting2023prophet,braun2024approximating,braverman2025new,naor2025online} and stationary settings \cite{aouad2020dynamic,kessel2022stationary,patel2024combinatorial,amanihamedani2025adaptive},
fair algorithms \cite{arsenis2022individual}, 
order-unaware algorithms
\cite{ezra2023next,sun2025online}, inflexible algorithms \cite{arnosti2022tight,perez2025iid}, and more.
These works drive home the message that approximating the optimum online algorithm using restricted algorithms is hard, 
but can often lead to better approximation than possible when comparing to the  (unattainable) benchmark of the ex-post optimum. 
We echo this message, showing that for our problem under i.i.d.~arrivals, a constant-approximation of the optimum online algorithm (using polytime algorithms) is possible, 
but is impossible when comparing to the optimum offline solution.


\subsection{Problem Formulation}\label{sec:prelims}

In the \emph{average-value-constrainted allocation problem (\AVA)}, allocating
item $i$ to buyer $j$ yields a value of $v_{ij}$.
Each buyer $j$ requires that the average value they obtain from allocated items
be at least $\rho_j$.
We wish to (approximately) maximize the total social welfare, or sum of values obtained by the buyers, captured by the following integer~LP:
\begin{align}
    \max & \sum_{(i,j)\in E}  v_{ij} \; x_{ij} \tag{\AVA-ILP} \label{AVA-ILP}
    \\ \textrm{s.t. } & \sum_i v_{ij} \; x_{ij}\geq \rho_j \cdot \sum_i x_{ij} \quad \qquad \forall \text{ buyer $j$} \nonumber \\
    &\sum_j x_{ij} \leq 1 \qquad \qquad \qquad \qquad\;\,\forall \text{ item $i$} \nonumber \\
     & x_{ij} \in \{0,1\} \qquad\qquad \qquad\qquad\,\; \forall \text{ items $i$, buyers $j$}. \nonumber
\end{align}
An instance $\calI$ of \AVA can be captured by a 
bipartite graph $(I,J,E)$, with a set $I$ of items and set $J$
of buyers, and edges $E \sse I\times J$, capturing all
buyer-item pairs with non-zero value. 
For $i \in I$ and $j \in J$,
edge $(i,j)$ has value $v_{ij}$.
We say edge $(i,j)$ is a \emph{\pedge} (positive edge) if it has
non-negative \emph{excess} $v_{ij} - \rho_{j} \geq 0$, and an
\emph{\nedge} otherwise, in which case we refer to $v_{ij}-\rho_{j}<0$
as its \emph{deficit}.  An item $i$ is a \emph{\pitem} if \emph{all}
its edges in $E$ are \pedges, and an \emph{\nitem} if \emph{all} its
edges in $E$ are \nedges: naturally, some items may be neither \pitems
or \nitems. We will call an instance \emph{unit-$\rho$} if $\rho_j=1$ for all buyers.\footnote{Such instances capture the core difficulty of the \AVA problem, and our examples (except those for \GAVA in Section~\ref{sec:hardness}) are unit-$\rho$ instances, so one can WLOG take $\rho_j=1$ in the first read.}

In the online setting, the $n$ buyers and their $\rho_j$ values are
known a priori, but items $i$ are revealed one at a time, together
with their value $v_{ij}$ for each buyer $j$, and an algorithm must
decide what buyer to allocate an item to (if any), immediately and
irrevocably on arrival. In the online i.i.d.~setting, $T$ items are
drawn (one after another) i.i.d.~from a known distribution over $m$
known item \emph{types}, with type $i$ (with value vector $(v_{i1},\dots,v_{in})$) drawn with probability
$q_i$. We say an edge type $(i,j)$ is an \emph{\nedge type} or a
\emph{\pedge type} if $v_{ij}-\rho_j<0$ or $v_{ij}-\rho_j \geq 0$,
respectively.




\subsection{Paper Outline}
We begin in \S\ref{sec:structure} by proving some structural lemmas
regarding \AVA, including an unintuitive non-linear dependence of the
welfare on the amount of supply. In
\S\ref{sec:matroid-GAP} we present the improved algorithm for the
offline setting giving \Cref{thm:offline}.  In \S\ref{sec:rounding-offline} we
present our LP-rounding algorithm for \AVA in an offline setting. We also discuss the approach's extendability, allowing to incorporate additional constraints, in \S\ref{sec:side-constraints}.
Building on this offline rounding-based algorithm, in \S\ref{sec:OPTon} we present a
constant-approximation of the optimum online algorithm. In the interest of space, we defer the discussion of competitive ratio bounds to \Cref{sec:OPToff}, and our hardness results to \Cref{sec:hardness}.


\section{The Structure of Near-optimal Solutions for \AVA}\label{sec:structure}

In this section, we show how to partition any feasible allocation of 
 \AVA instances (where we recall such allocation need not allocate all items) into structured subsets, which we call \emph{permissible bundles}.
This bundling-based structure will prove useful for all of our algorithms.

\begin{definition}[Bundling]
  A set $S$ of edges incident on buyer $j$ is a
  \emph{permissible bundle}~if
  \begin{enumerate}[nolistsep]
  \item $S$ consists of a single \pedge $(i^\star,j)$ and zero or
    more \nedges $(i,j)$, and 
  \item the edges in $S$ satisfy the average-value constraint, i.e., $\sum_{(i,j)\in S}
    v_{ij}\geq \rho_j\cdot|S|$. 
  \end{enumerate}
  A \emph{bundling-based} solution is one that
can be partitioned into a collection of permissible bundles.
  \label{def:permissible_bundle}
\end{definition}

Clearly, no bundling-based solution can be better than the best unconstrained solution, but in the following lemma we show a converse, up to constant factors. (Throughout, we use the shorthand notation $v \cdot x := \sum_{ij} v_{ij}x_{ij}$ for any vector $x\in \mathbb{R}^E$.)

\begin{lemma}[Good Bundling-Based Solution]
    \label{lem:bundling}
    Let $x^*$ be a solution to an instance of \AVA. Then, there exists a bundling-based solution
    $\widehat{x}$ of value at least $v\cdot \widehat{x}\geq \frac{1}{2}\; v\cdot x^*$.
\end{lemma}
As a corollary, the best bundling-based solution is a $2$-approximation, and so we will strive to approximate such bundling-based solutions.

We prove a strengthening of \Cref{lem:bundling} which also addresses online settings. 

\begin{definition}[Committed Bundling]
An online algorithm is a \emph{committed} bundling-based algorithm if its solution consists of
permissible bundles, and items can only be added to bundles; in particular, it \emph{commits}
to the allocation of each item to a particular bundle, and does not move items between permissible bundles.
\end{definition}

\begin{lemma}[Online Bundling-Based Solution]\label{lem:online-bundling}
Let $x^*$ be a solution to an instance of \AVA, with $x^*$ revealed online and (all interim partial solutions) satisfying the average-value constraints throughout. Then there exists a solution $\widehat{x}$ that is the
    output of a committed online bundling-based algorithm, of value at least  $v\cdot \widehat{x}\geq \frac{1}{2}\; v\cdot x^*$.
\end{lemma}

\begin{proof}
  For each buyer $j$, consider the edges
  $S := \{ (i,j) \mid x^*_{ij} = 1\}$ corresponding to items assigned
  to buyer $j$ in solution $x^*$, in order of addition to the solution $x^*$, namely $e_1,e_2,\dots,e_{|S|}$, with $e_k=(i_k,j)$.
  We now show how a committed online algorithm can output a collection 
  of permissible bundles of at least half the value from among the
  edges in $S$; doing this for each buyer proves the result. 

  Consider $i_k$, i.e., the $k$-th item allocated to $j$ by $x^*$, if $e_k$ is a \pedge (i.e. $v_{i_k,j}\geq \rho_j$), we denote $p=i_k$, open (create) a bundle $B_p =\{(j,p)\}$ and allocate appropriately in the new solution $\widehat{x}$.
  When $e_k=(i_k,j)$ is an \nedge, if $e_k$ can be added to some open bundle $B_p$ of $j$ while keeping it permissible, we add $(i_k,j)$ to $B_p$ in solution $\widehat{x}$; otherwise, we pick some open bundle $B_p$ of $j$ and mark it as closed (and never add more edges to this bundle). Since $x^*$ is feasible throughout the online arrival, for any $k\in[1,|S|]$ we have that $\sum_{\ell \leq k} v_{i_\ell, j} \geq k\cdot \rho_j$, and since we allocate all \pedges of $x^*$ in $\widehat{x}$ and only allocate a subset of the \nedges, we find that there must always be some open bundle of $j$ when considering an \nedge $e_k$.
  Therefore, the above (committed) bundling-based online algorithm is well-defined.
  Now, each bundle is closed by at most one \nedge $(i,j)$, and so we can charge the \nedges $(i,j)$ allocated in $x^*$ but not in $\widehat{x}$ to the \pedge $(p,j)$ in the bundle $B_p$ that they closed.
  But by definition of the \pedge and \nedge, we know $v_{pj}\geq \rho_{j} \geq v_{ij}$. 
  Therefore, denoting by $x^*_D$ the part of the solution $x^*$ that is \emph{discarded} in $\widehat{x}$ and by
  $x^*_p$ and $x^*_n$ the value of the \pedges and \nedges allocated by both $x^*$ and the new solution $\widehat{x}$, we have that 
  $v\cdot x^*_D \leq v\cdot x^*_p$. 
  Hence, 
  \begin{align}\label{key-bundling}
  v\cdot x^* & = v\cdot x^*_D + v\cdot (x^*-x^*_D) \leq 2\,v\cdot x^*_p + x^*_n \leq 2\,v\cdot (x^*_p + x^*_n).
  \end{align}
    That is, the obtained bundles of the solution
   $\widehat{x}=x^*_p+x^*_n$ constitute a $2$-approximation. 
\end{proof}

\begin{remark}\label{rem:bundling-2-tight}
This loss of a factor of two in the value is tight. To see this, consider a single-buyer unit-$\rho$ \AVA instance. There are $\frac{1}{\eps}$ \nedges each with value $1-\eps$ and $\frac{1}{\eps(1-\eps)}$ \pedges each with value $1+\eps(1-\eps)$. It is feasible to allocate all items to the buyer, and (arbitrarily close to) half the value of this solution is given by \nedges, but any permissible bundle contains no \nedges as any single \pedge doesn't have enough excess to cover the deficit of any \nedge.
\end{remark}

For our algorithms it will be convenient if each item is incident only on \pedges, or only on \nedges, thus removing the ambiguity about whether to use these as the single \pedge in a permissible bundle. 
Fittingly, we call such instances \emph{unambiguous}.
For example, when all buyers have the same average-value constraint (i.e. $\forall j: \rho_j=\rho$), for any item $i$ incident on a \pedge (i.e., $\exists j: v_{ij}\geq \rho$), we can trivially drop all \nedges of the item (i.e., drop $(i,j')$ where $v_{ij'}<\rho$) since there is no reason to allocate any \nedge instead of a \pedge of $i$,  and so making such instances unambiguous comes with no cost.
As we now show, any instance of \AVA in general can be made unambiguous while still preserving a bundling-based allocation that is constant-approximate for the original instance.

\begin{lemma}[Bundling Unambiguous Sub-Instances]
  \label{only-p-or-n}
  Given an \AVA instance $\calI = (I,J,E)$, dropping all of the \pedges or all the \nedges of each item $i\in I$ independently with probability $\nicefrac{1}{2}$ results in an unambiguous sub-instance
  $\calI' = (I,J,E')$ (where $E'\sse E$), admitting a bundling-based solution $x'$ which is $4$-approximate for $\calI$. 
\end{lemma}

\begin{proof}
    Let $x^*$ be an optimal solution for $\calI$. If we denote by $x^*_p$ and $x^*_n$ the characteristic vector for \pedges and \nedges allocated by both $x^*$ and $\widehat{x}=x^*_p + x^*_n$ as in the proof of \Cref{lem:online-bundling}, then, by the penultimate inequality of \Cref{key-bundling}, we have that $v\cdot x^* \leq 2\,v\cdot x^*_p + v\cdot x^*_n$.
    Now, consider the solution $x'$ consisting of all \pedges allocated in $\widehat{x}$ that were not dropped and all non-dropped \nedges allocated in bundle $S$ whose \pedge was also not dropped. We therefore have that this new solution has value precisely $\frac{1}{2}\,v\cdot x^*_p + \frac{1}{4}\,v\cdot x^*_n$, and so, by \Cref{key-bundling}, we have that $x'$ is a $4$-approximation, since
  \begin{align*}
  v\cdot x^* & \leq 4\cdot \left(\frac{1}{2}\,v\cdot x^*_p + \frac{1}{4}\,v\cdot x^*_n\right) = 4\, v\cdot x'. \qedhere
  \end{align*}
\end{proof}
We also provide an alternative, deterministic method to find such an unambiguous sub-instance. 
However, since our algorithms are randomized, we defer discussion of this method to \Cref{app:structure}. Note in unambiguous instances, every item is either a \pitem or an \nitem.

\subsection{Welfare is non-linear in supply}

In this section we provide a bound on the multiplicative gain in welfare in terms of increased supply. This will prove useful later. For now, it illustrates non-linearity of the \AVA problem in its supply. (This is in contrast to other allocation problems where the welfare is at best linear in the supply.)

To motivate this bound, consider the outcome of creating $k$ copies of each item in an \AVA instance. Clearly, the welfare increases by a factor of at least $k$, as we can just repeat the optimal allocation for the original instance $k$ times.
 However, as the following example illustrates, welfare can be \emph{super-linear} in the supply size increase for \AVA. 
\begin{example}\label{ex:nonlinear}
Consider a unit-$\rho$ instance of $k$-buyer \AVA with a single \pitem of value $1+k \eps$ for all buyers and $k$ many \nitems, with the $i$-th \nitems having value zero for all buyers except for one distinct buyer $i$, to whom it has value $1-\eps$.
In this instance $\mathsf{OPT}\approx 2$, since the \pitem can only be allocated to a single buyer, who can then only be allocated one \nitem, while in the instance obtained by creating $k$ copies of each item we can allocate a \pitem to each buyer together with $k$ many \nitems, and so for this instance $\mathsf{OPT}\approx k^2$, i.e., increasing supply $k$-fold increases the welfare \emph{$(k^2/2)$-fold}. 
\end{example}
The following lemma shows that the above example is an extreme case, and for a $k$-fold increase in supply, an $O(k^2)$-fold increase in welfare is best possible.

\begin{lemma}[Supply Lemma]\label{lem:parallel-repetition}
Let $\calI = (I,J,E)$ be an \AVA instance, and let $\calI' = (I',J,E')$ be an instance with the same buyer set and underlying costs and values obtained by copying each item in $\calI$ some $k$ times. 
$$\mathsf{OPT}(\calI') \leq O(k^2)\cdot \mathsf{OPT}(\calI).$$
\end{lemma}
\begin{proof}
Since bundling-based solutions are nearly optimal up to a constant factor of $2$, we can start with an optimal bundling-based allocation $\calA'$ for $\calI'$ and then independently associate each item of $\calI$ with a uniformly random copy among its $k$ copies in $\calI'$, allocating it as in $\calA'$. Finally, we remove all non-permissible obtained bundles to obtain allocation $\calA$ for $\calI$.
For each copy $i'$ of an item $i$, if  $i'$ is allocated in a \pedge in $\calA'$, the probability that $i$ is associated with $i'$ (and thus assigned to the same buyer by $\calA$) is precisely $1/k$. In contrast, if $i'$ is allocated in an \nedge by $\calA'$, the probability that $\calA$ allocates $i$ the same way as $i'$ is precisely $1/k^2$, as this requires both $i$ to be assigned to the same bundle (associated with the same copy) and the \pedge of this bundle to similarly be assigned to the same bundle. The lemma then follows by linearity of expectation. 
\end{proof}


\section{Offline Algorithm via Reduction to Matroid-Constrained GAP}\label{sec:matroid-GAP}

In this section we provide an improved constant-approximation for \AVA in the \emph{offline} setting; we will show in \Cref{sec:BBC-hard} that the problem is hard to approximate to better than $\frac{e}{e-1}$.

\begin{theorem}
\label{thm:rap-via-gap}
There exists a $(\frac{4e}{e-1}+o(1))$-approximate randomized algorithm for \AVA.
\end{theorem}

The algorithm proceeds by reducing \AVA to GAP with matroid constraints. 
Recall that an instance of the \emph{generalized assignment problem} (GAP)
consists of $n$ elements that can be packed into $m$ bins. Packing an
element $e$ into a bin $b$ gives a value $v_{eb}$ and uses up $s_{eb}$ space 
in that bin. If we let $y_{eb} \in \{0,1\}$ denote the
indicator for whether element $e$ is assigned to bin $b$, then
naturally $\sum_b y_{eb} \leq 1$. Each bin has unit size, and so the 
size of elements assigned to bin $b$ is at most $1$: in
other words, $\sum_e s_{eb} \; y_{eb}\leq 1$. The goal is to maximize the
total value of the assignment $\sum_{eb} v_{eb} \; y_{eb}$. 
\cite{fleischer2011tight} gave a $(1-1/e)$-approximation for this problem.
\cite{calinescu2011maximizing} gave the same approximation for an
extension of the problem, where the opened subset of
bins must be an independent set in some given matroid $\calM$.

\begin{theorem}\label{thm:umabiguous}
  There exists a randomized polynomial-time algorithm that, 
  for any unambiguous \AVA instance, outputs a solution with expected value at least $\big( 1 - \nf1e - o(1) \big)$ times the optimal bundling-based solution.
\end{theorem}

\begin{proof}
  Given an unambiguous \AVA instance  (i.e., one where each item is incident on only \pedges or only \nedges), we construct an instance of Matroid-Bin GAP 
  as follows:
  \begin{enumerate}
  
  \item \emph{Elements and bins:} 
  For each \pitem $p$ and buyer $j$, construct a bin $(p,j)$
    in the GAP instance. The elements of the GAP instance are exactly
    the items of the \AVA instance.
    
  \item \emph{Values/sizes of \pitems:} Assigning a \pitem $p$ to bin $(p,j)$ yields value
    $v_{pj}$ and uses zero space; Assigning \pitem $p$ to a bin
    $(p',j)$ with $p \neq p'$ yields value zero and uses $1+\eps$ space.
    
  \item \emph{Values/sizes of \nitems:} Assigning \nitem $i$ to bin $(p,j)$ yields value $v_{ij}$ 
    and uses $\frac{\rho_j - v_{ij}}{v_{pj} - \rho_j}$~space.
    
  \item \emph{Matroid on the bins:} Finally, the matroid $\calM$ on the bins is a partition matroid, 
    requiring that we choose at most one bin from $\{(p,j) \mid j \in B \}$, for each
    \pitem $p$. 
  \end{enumerate}
  The construction above results in a value-preserving one-to-one correspondence between feasible GAP solutions which are \emph{maximal}, i.e., where each \pitem $p$ is assigned to some bin, and permissible bundling-based solutions to the \AVA instance. 
  Indeed, for any feasible bundling-based solution to the \AVA instance, fix a bundle $(p,j)$ containing the item set $S$. 
  The value of placing the items in $S$ in the bin $(p,j)$ is precisely $\sum_{i\in S} v_{ij}$. Summing over all bins, we find that both solutions (to the \AVA and GAP instance) have the same value.
  On the other hand, the GAP solution is feasible since for each \pitem $p$ we open up at most one bin $(p,j)$ (thus respecting the matroid constraint) and moreover each bin's size constraint is respected due to the per-bundle average-value constraint and the zero size of $p$ in bin $(p,j)$, implying that
   $\sum_{i \in S} s_{i,(p,j)} = \sum_{i\in S\setminus \{p\}} \frac{\rho_j - v_{ij}}{v_{pj} - \rho_j}\leq 1.$ 
   Similarly, starting with a maximal solution to the GAP instance, the single bin $(p,j)$ into which $p$ is placed has its average-value constraint satisfied (note that $p$ cannot be placed in a bin $(p',j)$ for $p'\neq p$, where its size is $1+\eps$), and the value of the bundles obtained this way is the same as the GAP solution's value.
  Now the $(1-\nf1e-o(1))$-approximation algorithm for GAP with matroid constraints \cite{calinescu2011maximizing} gives the same approximation for \AVA on
    unambiguous instances. 
\end{proof}

\Cref{thm:umabiguous} combined with 
\Cref{only-p-or-n} completes the proof of \Cref{thm:rap-via-gap}.

\section{An Offline Algorithm via Relax-and-Round}
\label{sec:rounding-offline}

Let us now present an LP-rounding based algorithm for \AVA.
This more sophisticated algorithm yields another constant-approximate offline algorithm, which also allows to incorporate additional budget constraints (see \Cref{sec:side-constraints}). Moreover, this section's algorithm also provides a template for our main \emph{online} algorithms. 

The natural starting point for an LP-rounding based algorithm, the LP relaxation obtained by dropping the integrality constraints of \eqref{AVA-ILP}, turns out to be a dead end. This relaxation has an integrality gap of $\Omega(n)$ on $n$-buyer instances,\footnote{Recall that an LP relaxation's \emph{integrality gap} is the difference in objective between its best fractional and integral solutions.} even for unit-$\rho$, as shown by reinspecting the instance of \Cref{ex:nonlinear}.
\begin{example}
Consider an $n$-buyer unit-$\rho$ instance with a single \pitem $p$ of value $1+n\eps$ for all buyers, and $n$ \nitems, with the $i$-th \nitem having zero value for all buyers except for buyer $j_i$, for whom its value is $1-\eps$. An assignment $x_{pj}=\frac{1}{n}$ for all buyers $j$ and $x_{ij_i}=1$ for every \nitem $i$ gives value $n+1$ for the LP relaxation of \eqref{AVA-ILP}, while clearly the optimal integral solution has value $\approx 2$.
\end{example}
Therefore, to obtain any constant approximation via LP rounding, we need a tighter relaxation. 
To this end, we rely on \Cref{lem:bundling,only-p-or-n}, and provide the following relaxation for \emph{bundling-based} solutions for unambiguous \AVA instances.
This LP has decision variables $x_{ijp}$ for ($P$ or $N$)-item $i$, buyer $j$ and \pitem $p$. 
Informally, these correspond to the probability that $i$ is allocated to $j$ in the bundle with \pitem $p$, which we denote by $jp$. (Note: this polynomially-sized LP is clearly poly-time solvable.) 
\begin{align}
\max \quad & \sum_{i,j,p} v_{ij}\; x_{ijp} \tag{Bundle-LP} \label{bundle-LP} \\
\textrm{s.t.} \quad & \sum_i (\rho_{j} - v_{ij}) \; x_{ijp} \leq 0 \label{cons:RoS} & \forall j,p \\
& \sum_{j,p} x_{ijp} \leq  1 & \forall i \label{cons:item}\\
& x_{ijp} \leq x_{pjp} & \forall i,j,p \label{cons:bundle-defined-by-p-item}\\
& x_{p'jp}  = 0 & \forall j, P\textrm{-item }p' \neq p \label{cons:unique-p-item-per-bundle} \\
& x_{ijp}\geq 0 & \forall i,j,p 
\nonumber
\end{align}

Intuitively, the bundling, and in particular \Cref{cons:bundle-defined-by-p-item}, will allow us to overcome the integrality gap example above. We formalize this intuition later by approximately rounding this LP, but first we show that \eqref{bundle-LP} is a relaxation of bundling-based allocations for unambiguous \AVA instances.
\begin{lemma}\label{bundle-lp-gap}
For any unambiguous \AVA instance, the value of \eqref{bundle-LP} is at least as high as that of any bundling-based allocation.
\end{lemma}
\begin{proof}
Fix a (randomized) bundling-based allocation algorithm $\calA$. 
Let $Y_{ijp}$ be the indicator for $\calA$ having allocated item $i$ in bundle $jp$. 
We argue that $Y_{ijp}$ satisfy the constraints of \eqref{bundle-LP}, realization by realization.
Consequently, by linearity of expectation, so do their marginals, $\E[Y_{ijp}]$.
Constraint \eqref{cons:RoS} holds since $\calA$ satisfies the average-value constraint for each bundle. 
Constraint \eqref{cons:item} holds since each item is allocated at most once.
Constraint \eqref{cons:bundle-defined-by-p-item} holds because bundle $jp$ must be opened for $i$ to be allocated in it.
Constraint \eqref{cons:unique-p-item-per-bundle} holds since permissible bundles have a single \pitem in them.
Finally, non-negativity of $\bf{Y}$ is trivial.
We conclude that $\ex{\bf{Y}}$ is a feasible solution to the above LP, with objective precisely $\sum_{ijp} v_{ij} \; \E[Y_{ijp}]$. 
The lemma follows.
\end{proof}


We now turn to rounding this LP. To this end, we consider a two-phase algorithm, whose pseudo-code is given in \Cref{alg:offline-rounding}.
In Phase I  we \emph{open bundles}, letting each \pitem $p$ pick a single buyer $j$ with probability $x_{pjp}$,\footnote{Since Constraint \eqref{cons:item} is tight for every \pitem in any optimal LP solution, $\{x_{pjp}\}_j$ is a distribution over buyers.} and opening the bundle $jp$.
In Phase II we \emph{enrich the bundles}, by adding \nitems to them. Specifically, for each \nitem $i$, we create a set $S_i$ containing each open bundle $jp$ independently with probability $\alpha\cdot \frac{x_{ijp}}{x_{pjp}}$, where $\alpha\in[0,1]$ is a parameter to be specified later. Then, if this set $S_i$ contains a single bundle $jp$ and adding $i$ to this bundle would not violate the average-value constraint restricted to the bundle (denoted by $\mathsf{BundleAV}_{jp}$), i.e., this bundle would remain permissible, then we allocate $i$ to the bundle $jp$. Otherwise, we leave $i$ unallocated.

\begin{algorithm}
    \caption{Offline rounding of Bundle-LP}\label{alg:offline-rounding}
	\begin{algorithmic}[1]
        \State Make the instance unambiguous as in \Cref{only-p-or-n}
	\State Let $\mathbf{x}$ be an optimal solution to \eqref{bundle-LP} for the obtained unambiguous instance
	\For{\textbf{each} \pitem $p$} \Comment{Phase I}
	\State \label{line:open-bundle} Pick $j$ according to distribution $\left\{x_{pjp}\right\}_{j=1,\ldots,n}$ and open bundle $jp$
	\EndFor
	\For{\textbf{each} \nitem $i$} \Comment{Phase II}
    \State $S_i \gets \emptyset$
    \For{\textbf{each} bundle $jp$, \textbf{with probability} $\alpha \cdot \frac{x_{ijp}}{x_{pjp}}$} 
    \label{line:consider-bundle}
    \If{$jp$ was opened in Phase I}
    \State $S_i\gets S_i\cup \{jp\}$ 
    \EndIf
    \EndFor 
    \If{$|S_i|=1$}
    \If{the only bundle $jp\in S_i$ remains permissible after adding $i$ to it} \label{line:permissible-test}
    \State Allocate $i$ to $jp$
    \EndIf
    \EndIf
    \EndFor 
    \end{algorithmic}
\end{algorithm}

\Cref{alg:offline-rounding} clearly outputs a feasible allocation, since it only allocates \nitems $i$ to a bundle $jp$ if this would not violate the average-value constraint of the bundle, and hence by linearity the average-value constraint of the buyer remains satisfied. Moreover, the algorithm is well-defined; in particular, the probability spaces defined in lines \ref{line:open-bundle} and \ref{line:consider-bundle} are valid, by constraints \eqref{cons:item} for \pitem $p$, and \eqref{cons:bundle-defined-by-p-item} for triple $i,j,p$, respectively.
We turn to analyzing this algorithm's approximation ratio. For this, we will lower bound the probability of each item $i$ to be allocated in bundle $jp$ in terms of $x_{ijp}$. 

By \Cref{line:open-bundle}, each \pitem $p$ is assigned in bundle $jp$ precisely with probability $x_{pjp}$.
Consequently, the expected value \Cref{alg:offline-rounding} obtains from \pitems is precisely their contribution to the LP solution's value.
It remains to understand what value we get from \nitems.

\subsection{Allocation of \texorpdfstring{\nitems}{N-items}}

To bound the contribution of \nitems, we consider any tuple of \nitem $i$, buyer $j$ and \pitem $p$. Note that \nitem $i$ is assigned to bundle $jp$ if and only if all the four following events occur:
\begin{enumerate}
\item $\calE_1$: the event that bundle $jp$ is open, which happens with probability $x_{pjp}$.
\item $\calE_2$: the event that the $\Ber(\alpha \cdot \frac{x_{ijp}}{x_{pjp}})$ in \Cref{line:consider-bundle} comes up heads for $jp$. 
\item $\calE_3$: the event that $S_i \setminus \{jp\}  = \emptyset$.
\item $\calE_4$: the event that $jp$ would remain permissible if we were to add $i$ to bundle $jp$. 
\end{enumerate}
We note that events $\calE_1,\calE_2,\calE_3$ are all independent, as they depend on distinct (and independent) coin tosses. So, for example, $\pr{S_i\ni jp}=\pr{\calE_1\land \calE_2} = \pr{\calE_1}\cdot \pr{\calE_2}=\alpha\cdot x_{ijp}$.
Moreover, we have the following simple bound on $\pr{\calE_3}$.
\begin{lemma}\label{lem:few-bundles}
    $\pr{\bigwedge_{\ell=1}^3\calE_\ell} = \prod_{\ell=1}^3\pr{\calE_\ell} \geq (1-\alpha) \cdot \alpha\cdot x_{ijp}.$
\end{lemma}
\begin{proof}
The equality follows from independence of $\calE_1,\calE_2,\calE_3$. We therefore turn to lower bounding $\pr{\calE_3}$.
  Since $\Pr[X>0] = \Pr[X\geq 1] \leq \E[X]$ for any integer random variable~$X\geq 0$ by Markov's inequality, 
  \[ \Pr[ \overline{\calE_3}] = \Pr[\;|S_i\setminus\{jp\}| > 0] \leq \ex{ \left|S_i \setminus \{jp\}\right|} 
  = \sum_{p' \neq p} \sum_{j'} \alpha \cdot  x_{ij'p'} \leq \alpha, \] 
    where the second equality follows from $\Pr[S_i\ni j'p'] = \alpha\cdot x_{ij'p'}$ by the above, and the last inequality follows from Constraint \eqref{cons:item}.
    Since $\pr{\calE_1}\cdot \pr{\calE_2}=\alpha\cdot x_{ijp}$, the lemma follows.
\end{proof}
\textbf{A challenge.} 
    As noted above, $\calE_1,\calE_2,\calE_3$ are independent, resulting in a simple analysis for the probability 
    $\pr{\bigwedge_{\ell=1}^3\calE_\ell} = \prod_{\ell=1}^3\pr{\calE_\ell}$.
    Unfortunately, lower bounding $\Pr[\calE_4 \mid \calE_1 \land
  \calE_2 \land \calE_3]$ is more challenging, due to possible \emph{negative correlations} between $\calE_4$ and $\calE_3$. To see this, note that $\calE_1\land \calE_2 \land \calE_3$ implies $S_i=\{jp\}$. But this event can be positively correlated with other \nitems $i'$ having $S_{i'}=\{jp\}$, thus making it more likely that $jp$ won't be able to accommodate $i$ under $\mathsf{BundleAV}_{jp}$.
  
We can overcome this challenge of negative correlations, provided $(i,j)$ has small deficit compared to $(p,j)$'s excess. (We address the large deficit case separately later.)
  Specifically, by coupling our algorithm with an algorithm that allocates more often and does not suffer from such correlations, we can lower bound this conditional probability as follows.

\begin{lemma}\label{lem:small-excess-has-room}
Let $\beta\in [0,1]$.
If $i,j,p$ are such that $\rho_{j}-v_{ij}\leq \beta\cdot (v_{pj}-\rho_j)$, then 
$$\Pr[\calE_4 \mid \calE_1 \land
  \calE_2 \land \calE_3] \geq 1-\frac{\alpha}{1-\beta}.$$
\end{lemma}
\begin{proof}
Consider an imaginary algorithm $\calA'$ that allocates every \nitem $i'$ into every bundle $j'p'\in S_{i'}$,  even when $|S_{i'}|>1$ (so we may over-allocate $i'$) and even if this violates the $\mathsf{BundleAV}_{j'p'}$ constraint. 
Coupling $\calA'$ with \Cref{alg:offline-rounding} by using the same randomness for both algorithms, we have that item $i'$ is allocated to bin $j'p'$ by $\calA'$ with probability precisely $\pr{S_{i'}\ni j'p'}=\alpha\cdot x_{i'j'p'}$. 
In particular, $\calA'$ only allocates more items than \Cref{alg:offline-rounding}.

We denote by $N'_{jp}$ the set of \nitems allocated to bundle $jp$ by $\calA'$. Now, for $i$ as above, let $\calE'_4$ be the event that $\sum_{i'\in N'_{jp}\setminus\{i\}} (\rho_j - v_{i'j}) \leq (1-\beta)\cdot (v_{pj} - \rho_j)$, that is, the deficit of \nitems other than $i$ that $\calA'$ allocated to the bundle $jp$ together only consumes at most a $(1-\beta)$ fraction of $p$'s excess for $j$. By the small deficit assumption on $i,j,p$, we know that event $\calE'_4$ is sufficient for $\mathsf{BundleAV}_{jp}$ to be satisfied if \Cref{alg:offline-rounding} were to add $i$ to $jp$. Thus, $\calE'_4$ implies $\calE_4$ in any realization (of the randomness), since $\calA'$ only allocates more items to each bin than \Cref{alg:offline-rounding}.
On the other hand, we also have that both $\calE'_4$ and $\calE_1$ are independent of both $\calE_2\land \calE_3$, since the latter combined event depends on an independent random coin toss ($\calE_2$) and events concerning other bundles $j'p'$, which are both independent of the randomness concerning bundle $jp$. (Here we use that $\calA'$ allocates $i$ to $jp$ whenever $S_i\ni jp$, regardles of other bundles $j'p'$ belonging to $S_i$.) Consequently, by standard applications of Bayes' Law, we obtain the following.
\begin{align*}
\Pr[\calE'_4 \mid \calE_1 \land \calE_2 \land \calE_3] & = \Pr[\calE'_4 \mid \calE_1].
\end{align*}
As the imaginary algorithm $\calA'$ assigns $i'$ to $jp$ (i.e. $i'\in N'_{jp}$) iff $S_{i'}\ni jp$, we know that
\begin{align*}
\ex{\sum_{i'\in N'_{jp}} (\rho_j - v_{i'j}) \,\,\middle\vert\,\, \calE_1 } = \sum_{i'\neq p} (\rho_j - v_{i'j})\cdot \pr{S_{i'}\ni jp\mid \calE_1} = \alpha\cdot \sum_{i'\neq p} (\rho_j - v_{i'j})\; \frac{x_{i'jp}}{x_{pjp}} \leq \alpha \cdot (v_{pj} - \rho_j).
\end{align*}
Above, the second equality follows from linearity and $\Pr[S_{i'}\ni jp \mid \calE_1] = \alpha\cdot \frac{x_{ij'p'}}{x_{pjp}}$, and the inequality
follows from the average-value constraint for bundle $jp$ (i.e. \Cref{cons:RoS}) in our LP.
Therefore, by Markov's inequality 
\begin{align*}
\Pr\left[\sum_{i'\in N'_{jp}\setminus \{i\}} (\rho_j - v_{i'j}) > (1-\beta)\cdot (v_{pj} - \rho_j) \,\,\middle\vert\,\, \calE_1 \right] \leq \frac{\ex{\sum_{i'\in N'_{jp}\setminus \{i\}} (\rho_j - v_{i'j})\,\,\middle|\,\, \calE_1}}{(1-\beta)\cdot (v_{pj} - \rho_j)}\leq \frac{\alpha}{1-\beta}, 
\end{align*}
and thus $\Pr\left[\calE'_4\mid \calE_1\right]\geq 1-\frac{\alpha}{1-\beta}$. Recalling that $\calE'_4$ implies $\calE_4$ in any realization, we conclude with the desired bound, as follows.
\begin{align*}
\Pr[\calE_4 \mid \calE_1 \land
  \calE_2 \land \calE_3] & \geq \Pr[\calE'_4 \mid \calE_1 \land
  \calE_2 \land \calE_3] = \Pr[\calE'_4 \mid \calE_1] \geq 1-\frac{\alpha}{1-\beta}. \qedhere
\end{align*}
\end{proof}

\Cref{lem:small-excess-has-room} and the preceding discussion yield a lower bound on the probability of an \nitem $i$ being successfully allocated to a bundle $jp$ when $i$'s deficit is small relative to the excess of the \pitem $p$.
For the large deficit case, no such bound holds. However, as we now observe (see proof in \Cref{app:offline}), large-deficit edges contribute a relative small portion of the allocation's value in the optimal LP solution.
\begin{restatable}{lemma}{largeExcessNotWorthwhile}\label{lem:large-excess-not-worthwhile}
Let $\beta\in [0,1]$. For any bundle $jp$, let $L^{\beta}_{jp}$ denote the set of \emph{$\beta$-large deficit} \nitems for bundle $jp$, i.e., \nitem $i$ with $\rho_j - v_{ij} > \beta\cdot (v_{pj} - \rho_j)$.
Then, 
$$\sum_{j,p}\sum_{i\in L^{\beta}_{jp}} v_{ij}\; x_{ijp}\leq \frac{1}{\beta} \; \sum_{j,p} v_{pj}\; x_{pjp}.$$
\end{restatable}

\subsection{Completing the analysis}

We are now ready to bound the approximation ratio of \Cref{alg:offline-rounding}.

\begin{theorem}\label{thm:offline-rounding}
\Cref{alg:offline-rounding} with $\alpha=0.3$ yields a $32$-approximation for \AVA.
\end{theorem}
\begin{proof}
Let $\beta\in[0,1]$ be some constant to be determined and let $\gamma = \gamma(\alpha,\beta):=\alpha \cdot (1-\alpha)\cdot \left(1-\frac{\alpha}{1-\beta}\right)$. 
Denote $N_{jp}$ by the set of \nitems allocated to bundle $jp$ by the algorithm. 
By \Cref{lem:small-excess-has-room,lem:few-bundles} we have for bundle $jp$ and \nitem $i\notin L^{\beta}_{jp}$ that
\begin{align*}
\Pr[i\in N_{jp}] = 
\pr{\calE_4 \,\,\middle|\,\, \bigwedge_{\ell=1}^{3} \calE_\ell}
\pr{\bigwedge_{\ell=1}^3\calE_\ell}
%
& \geq \left(1-\frac{\alpha}{1-\beta}\right)\cdot \alpha\cdot (1-\alpha)\cdot x_{ijp} = \gamma\cdot x_{ijp}.
\end{align*}

Therefore, by linearity of expectation and \Cref{lem:large-excess-not-worthwhile}, the expected value of the (feasible) random allocation of \Cref{alg:offline-rounding} is at least 
\begin{align*}
    & \sum_{j,p} v_{pj}\; x_{pjp} + \gamma\sum_{i,j,p: i\neq p} v_{ij}\; x_{ijp} - \gamma \sum_{j,p}\sum_{i\in L^{\beta}_{jp}} v_{ij} \; x_{ijp} \geq \left(1-\frac{\gamma}{\beta}\right)\; \sum_{j,p} v_{pj}\; x_{pjp} + \gamma \; \sum_{i,j,p: i\neq p} v_{ij}\; x_{ijp}. 
\end{align*}
So, this algorithm's output has value at least a $\min\{1-\frac{\gamma}{\beta},\;\gamma\}$ fraction of the optimal LP value; i.e., it is a $1/\min\{1-\frac{\gamma}{\beta},\;\gamma\}$-approximation.
Taking $\alpha \approx 0.3$ and $\beta\approx 0.156$ (optimized by an off-the-shelf numerical solver) yields a ratio of $1/0.13<8$.
The theorem then follows from \Cref{bundle-lp-gap} and \Cref{only-p-or-n}.
\end{proof}

\subsection{Extension: adding budget constraints}\label{sec:side-constraints}
Before moving on to our online algorithms, we note that the LP-based approach 
has the benefit of allowing to incorporate (some) additional constraints seamlessly. 
To illustrate this, we show in 
\Cref{app:offline} that our LP and algorithm, with minor modifications, allow to 
approximate allocation problems with both the average-value constraint and $O(1)$ many budget constraints (for every buyer), corresponding to different resources.
More formally, for a cost function $\ell$ (e.g., corresponding to storage, time, or other costs), each buyer $j$ has some budget $B^{(\ell)}_j$, and the $\ell$-cost of allocation to buyer $j$ must not exceed this budget. 
That is, for $x_{ij}\in \{0,1\}$ an indicator for item $i$ being allocated to buyer $j$, we have
\begin{gather}
 \forall j,\;\;\; \ell\text{-cost}_j = \sum_{i} \ell_{ij} \; x_{ij} \leq B^{(\ell)}_{j}. \label{eq:budget}
\end{gather}
The \emph{small-cost} assumption (a.k.a.~the small-bids assumption for online AdWords~\cite{mehta2007adwords}) stipulates that no particular item has high cost compared to the budget, i.e. $\max_{ij} \ell_{ij}/B^{(\ell)}_j \leq \eps \to 0$. This is motivated by the motivating application for online AdWords, where buyers (advertisers) typically provision a large enough budget to allow them to pay for many ads to be displayed (i.e., for many items to be allocated to the buyer).

\begin{restatable}{theorem}{multiconstraint}\label{thm:multiconstant}
There exists a constant-approximate algorithm for \AVA and any constant number of budget constraints (for every buyer) subject to the small-bids assumption.
\end{restatable}
The same arguments in this section extend to our online algorithms, but are omitted for brevity.
\color{black}

\section{Online  Algorithms: Approximating the Online Optimum}\label{sec:OPTon}

In this section and the next we study \AVA in the online i.i.d.~setting (see \Cref{sec:prelims} for definition and notation). 
Specifically, in this section we provide a polynomial-time online algorithm which provides a constant approximation of the optimal online algorithm.

First, by \Cref{lem:online-bundling}, we have that the optimal online algorithm is approximated within a factor two by a
bundling-based online algorithm which is \emph{committed}.
As we will show, the following LP provides a relaxation for the value of the best such online algorithm.
Our LP consists of variables $x_{ijp}$ for each item type $i\in[m]$, buyer $j\in [n]$ and item type $p$ such that $(p,j)$ is a \pedge.
\begin{align}
\max \quad & \sum_{i,j,p} v_{ij}\; x_{ijp} \tag{OPTon-Bundle-LP}
\label{opton-bundle-LP} \\
\textrm{s.t.} \quad & \sum_i (\rho_{j} - v_{ij}) \; x_{ijp} \leq 0 \label{opton-cons:RoS} & \forall \text{ \pedge type } (p,j)  \\
& \sum_{j,p} x_{ijp} \leq  q_i\cdot T & \forall \textrm{ item type }i \label{opton-cons:item}\\
& x_{ijp} \leq x_{pjp}\cdot q_i\cdot T & \forall \textrm{ \nedge type } (i,j), \textrm{\pedge type }(p,j) \label{opton-cons:bundle-defined-by-p-item}\\
& x_{p'jp}  = 0 & \forall \textrm{ \pedge types }(p,j) \neq (p',j) \label{opton:unique-p-item-per-bundle} \\
& x_{ijp}\geq 0 & \forall \textrm{ item type }i, \textrm{ \pedge type }(p,j) \nonumber
\end{align}
\begin{lemma}\label{lem:opton-LP}
\eqref{opton-bundle-LP} has value which is at least half the expected value of any online \AVA algorithm under i.i.d.~arrivals (from the same distribution used in the LP), where item type $i$ is drawn with probability
$q_i$.
\end{lemma}
\begin{proof}
First, by the Online Bundling Lemma (\Cref{lem:online-bundling}),
the best committed online bundling-based algorithm 2-approximates the best online algorithm.
We therefore turn to showing that \eqref{opton-bundle-LP} is a relaxation of the value of the best committed bundling-based online algorithm, $\calA$.
Let $x_{ijp}$ be the average number of times a copy of item type $i$ is allocated in a copy of bundle $jp$ by $\calA$. 
Constraint \eqref{opton-cons:RoS} follows by linearity of expectation, together with the fact that each opened copy of bundle $jp$ must satisfy the average-value constraint.
Constraint \eqref{opton-cons:item} simply asserts that $i$ is allocated at most as many times as it arrives. Constraint \eqref{opton-cons:bundle-defined-by-p-item} holds for a committed online algorithm (that guarantees feasibility with probability $1$), for the following reason: for every copy of bundle $jp$ opened, no items can be placed in that bundle before it is opened. But the expected number of copies of $i$ to be assigned after any bundle $jp$ is opened is at most the number of arrivals of $i$ after this bundle is opened and is at most $q_i\cdot T$, which upper-bounds the ratio between $x_{ijp}$ and $x_{pjp}$.
All other constraints hold similarly to their counterparts in the proof of \Cref{bundle-lp-gap}.
\end{proof}

\textbf{Note:} Constraint \eqref{opton-cons:bundle-defined-by-p-item} is reminiscent of constraints bounding the optimal online algorithm in the secretary problem literature \cite{buchbinder2014secretary} and prophet inequality literature \cite{papadimitriou2021online}.

The outline of our algorithm is similar to that of \Cref{alg:offline-rounding}, though as it does not have random access to the different items throughout, it first allocates \pedges in the first $T/2$ arrivals, and only then allocates \nedges in the last $T/2$ arrivals.
To distinguish between bundles opened at different times, we now label copies of bundle type $jp$ (i.e., items allocated to buyer $j$ with single \pedge of type $(p,j)$) opened at time $t$ by $jpt$. 
The algorithm's pseudocode is given in \Cref{alg:online-rounding}. 

Note that in our online algorithms (here and in~\Cref{sec:OPToff}), the LPs are based on distributions that can be ambiguous in the sense that each item type in the distribution can have both \pedges and \nedges, and we don't explicitly modify the distribution to make it unambiguous. However, our algorithm effectively makes each realized instance (of $T$ sampled items) unambiguous, as we ignore all \nedges incident to the first $T/2$ items and vice versa for the last $T/2$ items.

\begin{algorithm}[H]
    \caption{Online rounding of bundling-based LP}\label{alg:online-rounding}
	\begin{algorithmic}[1]
	\State Let $\mathbf{x}$ be an optimal solution to \Cref{opton-bundle-LP}
	\ForAll{arrivals $t=1,\dots,T/2$, of type $p$}
	\State Pick a $j$ according to the distribution $\{\frac{x_{pjp}}{q_p\cdot T}\}_{j=1,\dots,n}$ and open bundle $jpt$  
	\EndFor
        \ForAll{arrivals $t^\star =T/2+1,\dots,T$ of type $i$}
    \State $S_{it^\star} \gets \emptyset$
    \ForAll{bundles $jpt$, \textbf{with probability} $\frac{\alpha \cdot x_{ijp}}{x_{pjp}\cdot q_i\cdot T}$}  \label{line:consider-bundle-opton} 
    \If{bundle $jpt$ is open} 
    \State $S_{it^{\star}}\gets S_{it^{\star}} \cup \{jpt\}$
    \EndIf
    \EndFor 
    \If{$|S_{it^{\star}}|=1$}
    \If{$jpt\in S_{it^{\star}}$ remains permissible after adding $it^{\star}$ to it} \label{line:permissible-test-opton}
    \State Allocate $it^{\star}$ to $jpt$
    \EndIf
    \EndIf
    \EndFor
    \end{algorithmic}
\end{algorithm}

\subsection{Analysis}
In what follows we provide a brief overview of the relevant events in the analysis of \Cref{alg:online-rounding}, deferring proofs reminiscent of the analysis of \Cref{alg:offline-rounding} to \Cref{app:opton}.

First, the value obtained from \pedges by \Cref{alg:online-rounding} is clearly half that of the LP, by linearity of expectation.
In particular, we create $x_{pjp}/2$ copies of bundle $jp$ in expectation.
The crux of the analysis is in bounding our gain from \nedges.

To bound the contribution of \nedges, we note that a copy of item $i$ at time $t^{\star}>T/2$, which we denote by $it^{\star}$, is assigned to bundle $jpt$ if and only if all the five following events (overloading notation from \Cref{sec:rounding-offline}) occur:
\begin{enumerate}
\item $\calE_0$: the event that $it^{\star}$ is the realized item at time $t^{\star}$, which happens with probability $q_i$.
\item $\calE_1$: the event that bundle $jpt$ is open, which happens with probability $q_p\cdot \frac{x_{pjp}}{q_p\cdot T} = \frac{x_{pjp}}{T}$.
\item $\calE_2$: the event that the $\Ber(\frac{\alpha \cdot x_{ijp}}{x_{pjp}\cdot q_i\cdot T})$ in \Cref{line:consider-bundle-opton} comes up heads for $jpt$. 
\item $\calE_3$: the event that $S_{it^\star} \setminus \{jpt\}  = \emptyset$. 
\item $\calE_4$: the event that $jpt$ would remain permissible if we were to add $it^\star$ to bundle $jpt$. 
\end{enumerate}

Similarly to the events we studied when anlyzing our offline \Cref{alg:offline-rounding}, the events $\calE_0,\calE_1,\calE_2$ are independent, as are the events $\calE_1,\calE_2,\calE_3$. However, $\calE_3$ is not independent of $\calE_0$ (in particular, it occurs trivially if $\calE_0$ does not). Nonetheless, bounding $\pr{\bigwedge_{\ell=0}^3 \calE_\ell}$ is not too hard. The following lemma, whose proof essentially mirrors that of \Cref{lem:few-bundles}, and is thus deferred to \Cref{app:opton}, provides a bound on the probability of all first four events occurring.
\begin{restatable}{lemma}{simpleOPTonBound}\label{lem:simple-online-bound}
$\Pr[ \calE_0 \land \calE_1 \land \calE_2 \land \calE_3 ] \geq \alpha\cdot (1-\alpha/2)\cdot \frac{x_{ijp}}{T^2}$.
\end{restatable}

As with our offline \Cref{alg:offline-rounding}, the challenge in the analysis is due to possible negative correlations between $\calE_4$ and $\calE_3$. Similarly, we overcome this challenge of negative correlations, provided $(i,j)$ has small deficit compared to $(p,j)$'s excess, by coupling with an algorithm with no such correlations. 
(We address large-deficit $(i,j)$ later.) 
The obtained syntactic generalization of \Cref{lem:small-excess-has-room}, whose proof is deferred to \Cref{app:opton}, is the following.
  
\begin{restatable}{lemma}{optonroom}\label{lem:opton-small-excess-has-room}
Let $\beta\in [0,1]$.
If $i,j,p$ are such that $\rho_{j}-v_{ij}\leq \beta\cdot (v_{pj}-\rho_j)$, then 
$$\Pr[\calE_4 \mid \calE_0 \land \calE_1 \land
  \calE_2 \land \calE_3] \geq 1-\frac{\alpha}{2(1-\beta)}.$$
\end{restatable}

\Cref{lem:opton-small-excess-has-room} and the preceding discussion yield a lower bound on the probability of a copy of item $i$ be allocated to a bundle $jpt$ at time $t^\star$ if $i,j,p$ is in the small deficit case as the above lemma.
For large-deficit items, no such bound holds. However, large-deficit edges contribute a small portion of the allocation's value.
Specifically, \Cref{lem:large-excess-not-worthwhile}, holds for \eqref{opton-bundle-LP} as well, since the only constraint that this lemma's proof relied on was Constraint \eqref{cons:RoS}, which is identical to Constraint \eqref{opton-cons:RoS} in \eqref{opton-bundle-LP}.

We are now ready to bound the approximation ratio of \Cref{alg:offline-rounding}.

\begin{theorem}\label{thm:online-rounding}
\Cref{alg:online-rounding} with $\alpha=0.64$ is a polynomial-time algorithm achieving a $57$-approximation of the optimal online algorithm for \AVA  under known i.i.d.~arrivals.
\end{theorem}
\begin{proof}
That the algorithm runs in polynomial time follows from its description, together with the LP \eqref{opton-bundle-LP} having polynomial size (in the distribution size).
The analysis is essentially identical to that of \Cref{thm:offline-rounding}, with the following differences. First, we recall that the expected number of copies of bundle $jp$ opened is $\frac{T}{2}\cdot q_p\cdot \frac{x_{pjp}}{q_p\cdot T} = \frac{1}{2}\; x_{pjp}$. Next, by \cref{lem:simple-online-bound,lem:opton-small-excess-has-room}, the probability that copy $it^\star$ of small-deficit item $i$ for bundle $jpt$ is allocated to it is at least $\gamma\cdot \frac{x_{ijp}}{T^2}$, for $\gamma = \gamma(\alpha,\beta) := \frac{\alpha}{2}\cdot \left(1-\frac{\alpha}{2}\right)\cdot \left(1-\frac{\alpha}{2(1-\beta)}\right)$. 
Again, linearity of expectation and summation over all $(t,t^\star)\in [T/2]\times (T/2,T]$ in combination with \Cref{lem:large-excess-not-worthwhile} implies that for any $\beta\in [0,1]$, the gain of \Cref{alg:online-rounding} is at least 
\begin{align*}
    & \frac{1}{2}\left(\sum_{j,p} v_{pj}\; x_{pjp} + \frac{\gamma}{4} \sum_{i,j,p: i\neq p} v_{ij}\; x_{ijp} - \frac{\gamma}{4} \sum_{j,p}\sum_{i\in L^\beta_{jp}} v_{ij} \; x_{ijp}\right) \\
    \geq & \left(\left(\frac{1}{2}-\frac{\gamma}{4\beta}\right)\; \sum_{j,p} v_{pj}\; x_{pjp} + \frac{\gamma}{4} \; \sum_{i,j,p: i\neq p} v_{ij}\; x_{ijp}\right). 
\end{align*}
Therefore, by \Cref{lem:opton-LP}, \Cref{alg:online-rounding} yields a $2/\min\{\frac{1}{2}-\frac{\gamma}{4\beta},\;\frac{\gamma}{4}\}$-approximation.
This expression is optimized by $\alpha \approx 0.64$ and $\beta\approx 0.0766$, yielding a ratio of $\approx \frac{2}{0.0355}<57$, as claimed.
\end{proof}

\paragraph{Acknowledgments.} We thank the anonymous reviewers for useful comments and suggestions, which improved the paper's presentation.

{\small
\bibliographystyle{alpha}
\bibliography{abb,bib}
}

\appendix
\section{Online Algorithms: Approximating the Offline Optimum}\label{sec:OPToff}

In this section we look at the lower and upper bounds of the competitive ratio for online algorithms, i.e. the approximation of the ex-post optimum allocation's value, and we consider both the adversarial and i.i.d. cases.

\paragraph{Adversarial arrival.} In this setting, we note that no online algorithm can be $o(T)$-competitive. To see this, consider the unit-$\rho$ instance where the first $T-1$ arriving items have value $1-\eps$ for all $n=T$ buyers, followed by a single item at the end with value $1+ \eps T$ for a single adversarially chosen buyer and value $0$ for all other buyers. Any online algorithm cannot allocate any of the first $T-1$ items due to the average-value constraint, and thus can only get value $1+\eps T$ from the last item. In contrast, the ex-post optimum can allocate all items to one buyer and collect value $T+1-\eps$. On the other hand, a competitive ratio of $T$ is trivial to achieve for online \AVA, by simply allocating any 
item $i$ with a \pedge $(i,j)$ greedily to the buyer $j$ yielding the highest value. This is a feasible allocation and has value equal to the highest-valued edge in the $T$-item instance, which is obviously at least a $1/T$ fraction of the optimal allocation's value.\\

The rest of this section will therefore be dedicated to \AVA with i.i.d.~arrivals, as in~\Cref{sec:OPTon}, but now focusing on approximating the ex-post optimum. 
We start with the following result lower bounding the competitive ratio.

\begin{lemma}\label{lem:online-iid-hard}
There exists a family of uniform online i.i.d.~unambiguous unit-$\rho$ \AVA instances with $n=m=T\geq 2$ growing, 
on which every online algorithm's approximation ratio of the ex-post optimum is at least $\Omega\left(\frac{\ln n}{\ln \ln n}\right)=\Omega\left(\frac{\ln m}{\ln \ln m}\right)=\Omega\left(\frac{\ln T}{\ln \ln T}\right)$.
\end{lemma}
\begin{proof}
Let $\eps=\frac{1}{T}$.
Consider an instance with $T$ buyers $j_1,\dots,j_T$, where all buyers have $\rho=1$, and $T$ item types. Each item type $i\in [T-1]$ is an \nitem, with value $1-\eps$ for buyer $j_i$ and value zero for all others. (So, buyer $j_T$ has zero value for all \nitems.) The single \pitem type $T$ has value $1+\eps T$ for all buyers.
The $T$ arrival types are drawn uniformly from these $T$ types, and consequently there is a single arrival of each type in expectation.
Now, an online algorithm (that guarantees average-value constraints in any outcome) can only allocate \nitems to a buyer after the buyer was allocated a \pitem. But since each \nitem appears only once in expectation (and hence at most once after the arrival of a \pitem type), each allocation of a \pitem (and \nitems) to a buyer yields expected value at most $1+\eps T + 1 - \eps =  3-\eps$ to an online algorithm. Since only one \pitem arrives in expectation, an online algorithm accrues value at most $3-\eps$ in expectation on this instance family.

In contrast, the event $\calE$ that a single \pitem arrived satisfies $\Pr[\calE] = T\cdot \frac{1}{T}\cdot (1-\frac{1}{T})^{T-1} \geq (1-\frac{1}{T})^{T} \geq \frac{1}{4}$. Conditioned on $\calE$, we have a multi-nomial distribution for the number of arrivals $A_i$'s of the \nitem types. 
Therefore, by standard anti-concentration arguments for the classic balls and bins process~\cite{AzarBKU99}, we have
\begin{align*}\pr{\max_i A_i \geq \frac{\ln T}{\ln \ln T} - 1 \,\,\middle|\,\, \calE} = 1-o(1).
\end{align*}
Consequently, the offline algorithm which, if event $\calE$ occurs, allocates the single \pitem and all copies of $i^\star := \arg\max_i A_i$ to $j_{i^\star}$ yields expected value at least $\E[\max_i A_i \mid \calE] \cdot \Pr[\calE] = \Omega\left(\frac{\ln T}{\ln \ln T}\right)$.
Consequently, this asymptotic ratio also lower bounds any online algorithm's approximation ratio of the ex-post optimum. The full lemma statement follows, since $n=m=T$.
\end{proof}

\subsection{A matching algorithm assuming constant expected arrivals}
\Cref{lem:online-iid-hard} relied on anti-concentration. 
If the expected number of arrivals $A_i$ of each item type $i$ is at least some constant $\Gamma>0$,
 namely $\ex{A_i}=q_i\cdot T \geq \Gamma$ 
  (e.g., in  \Cref{lem:online-iid-hard} we had $q_i\cdot T=1$ for every $i$), then this anti-concentration is tight. 
  In particular, we have the following, by standard Chernoff bounds and union bound (see \Cref{app:optoff} for proof).
\begin{restatable}{observation}{Aiconcentrated}\label{Ai-concentrated}
If $\ex{A_i}\geq \Gamma$ for all $i\in [m]$ and $\kappa:=\frac{6}{\min(1,\;\Gamma)}\cdot \frac{\ln T}{\ln \ln T}$, then $$\pr{\max_i A_i \geq \kappa \cdot q_i\cdot T} \leq \frac{1}{ T^2}.$$
\end{restatable}

We will show that if the distribution satisfies the assumption on all $\ex{A_i}\geq \Gamma=\Theta(1)$, we can show an asymptotically matching upper-bound $O(\frac{\ln T}{\ln \ln T})$ of the competitive ratio. 

Our first ingredient towards this proof will, naturally, be another LP, this time capturing possible anti-concentration of arrivals.
Similar to~\eqref{opton-bundle-LP}, the LP has one variable $x_{ijp}$ for each item type $i\in[m]$, buyer $j\in [n]$ and item type $p$ such that $(p,j)$ is a \pedge.
\begin{align}
\max \quad & \sum_{i,j,p} v_{ij}\; x_{ijp} \tag{OPToff-Bundle-LP} \label{optoff-bundle-LP} \\
\textrm{s.t.} \quad & \sum_i (\rho_{j} - v_{ij}) \; x_{ijp} \leq 0 & \forall \text{ \pedge type } (p,j)   \label{optoff-cons:RoS} \\
& \sum_{jp} x_{ijp} \leq 2\cdot \lceil q_i\cdot T\rceil & \forall \textrm{ item type }i  \label{optoff-cons:nitem}\\
& x_{ijp} \leq x_{pjp}\cdot \lceil q_i\cdot T\cdot \kappa\rceil & \forall \textrm{ \nedge type } (i,j), \textrm{\pedge type }(p,j) \label{optoff-cons:bundle-defined-by-p-item}\\
& x_{p'jp}  = 0 & \forall \textrm{ \pedge types }(p,j) \neq (p',j) \label{optoff:unique-p-item-per-bundle} \\
& x_{ijp}\geq 0 & \forall \textrm{ item type }i, \textrm{ \pedge type }(p,j) \nonumber
\end{align}
\begin{restatable}{lemma}{optoffLP}
\label{lem:optoff-LP}
Fix an \AVA instance with i.i.d.~arrivals satisfying $q_i\cdot T\geq \Gamma=\Theta(1)$ for all~$i\in[m]$. Let $\mathsf{OPT}$ be the ex-post optimal value and let $V[\mathsf{OFF}]$ be the value of \eqref{optoff-bundle-LP}. Then, 
$$\ex{\mathsf{OPT}} \leq O(V[\mathsf{OFF}]).$$ 
\end{restatable}
\begin{proof}
By~\Cref{bundle-lp-gap}, we can restrict to the optimal ex-post bundling-based solution and just lose a factor of $2$ in the approximation ratio. We start with a trivial upper-bound on the value of $\mathsf{OPT}$ in any outcome of the i.i.d.~arrivals. Consider the instance with exactly one copy of each item type from the support of the distribution. The best bundling-based offline solution for this instance is upper-bounded by~\eqref{bundle-LP} (\Cref{bundle-lp-gap}), and this value is clearly upper bounded by $V[\mathsf{OFF}]$ since the constraints for \eqref{bundle-LP} are tighter than those of \eqref{optoff-bundle-LP}.
Under $T$ i.i.d. arrivals, each item can appear at most $T$ times, and thus by the Supply Lemma (\Cref{lem:parallel-repetition}) applied to the instance with a single occurrence per item type, we find that the following bound holds deterministically.
\[
\mathsf{OPT} \leq O(T^2)\cdot V[\mathsf{OFF}].
\]

Next, let $\calE$ be the event that no item type $i$ has more than $\lceil q_i\cdot T\cdot \kappa\rceil$ arrivals.
	By \Cref{Ai-concentrated}, $\pr{\calE}\geq 1-\frac{1}{T^2}$.
 Conditioned on $\calE$, consider the expected number of times (over the randomness of the i.i.d. arrivals) that the ex-post optimal bundling-based solutions allocate an item of type $i$ to a copy of bundle $jp$, and denote this value by $x_{ijp}$. We will argue that such $x_{ijp}$'s form a feasible solution for~\eqref{optoff-bundle-LP}. Since the expected value of the ex-post optimal bundling-based solution conditioned on $\calE$ is simply $\sum_{i,j,p}v_{ij}\;x_{ijp}$, this immediately gives that 
 \[
 \ex{\mathsf{OPT}\mid\calE}\leq 2\cdot V[\mathsf{OFF}].
 \]

 The proof that $x_{ijp}$ constructed above is feasible follows essentially the same argument as~\Cref{lem:opton-LP}. The average-value constraint~\eqref{optoff-cons:RoS} holds by linearity of expectation because the ex-post (bundling-based) optimum for any outcome satisfies the average-value constraint. Constraint~\eqref{optoff-cons:nitem} holds since the expected times we allocate items of type $i$ cannot exceed $i$'s expected number of occurrences, which is bounded by 
 $\ex{A_i \mid \calE} \leq \frac{\ex{A_i}}{\pr{\calE}} \leq \frac{q_i\cdot T}{1-1/T^2} \leq 2\cdot q_i\cdot T\leq  2\cdot \lceil q_i\cdot T \rceil.$
 Constraint~\eqref{optoff-cons:bundle-defined-by-p-item} holds since whenever a bundle $jp$ is opened in the ex-post optimum for any outcome, conditioned on $\calE$ we have at most $q_i\cdot T\cdot \kappa$ items of type $i$, which is a trivial upperbound on how many items of type $i$ can be allocated to bundle $jp$, and thus cap the ratio between $x_{ijp}$  and $x_{pjp}$.

Combining the above arguments together with linearity of expectation, the lemma follows.
\begin{align*}
\ex{\mathsf{OPT}} & =\ex{\mathsf{OPT}|\calE}\cdot\pr{\calE}+\ex{\mathsf{OPT}|\overline{\calE}}\cdot\pr{\overline{\calE}}\leq O(V[\mathsf{OFF}]). \qedhere
\end{align*}
\end{proof}

We make the simple observation that the two LPs~\eqref{opton-bundle-LP} and~\eqref{optoff-bundle-LP} only differ at the RHS of the constraints, with the most crucial difference being in the constraints upper bounding
$x_{ijp}/x_{pjp}$, where they differ by a factor of $\frac{\lceil q_i\cdot T\cdot \kappa\rceil}{q_i\cdot T}=O(\kappa)$ (using that $\Gamma=\Omega(1)$). As we prove in \Cref{app:optoff},  scaling down any feasible solution of the latter LP by $O(\kappa)$ yields a feasible solution to the former LP, leading to the following observation.
\begin{restatable}{observation}{optonoff}
\label{lem:optonoff}
Fix an \AVA instance with i.i.d. arrivals, satisfying $q_i\cdot T\geq \Gamma=\Theta(1)$ for all item type $i$. Then, $V[\mathsf{OFF}]$ and $V[\mathsf{ON}]$, the values of \eqref{optoff-bundle-LP} and \eqref{opton-bundle-LP} (respectively) satisfy
\[
V[\mathsf{OFF}]\leq O\left(\frac{\ln T}{\ln \ln T}\right)\cdot V[\mathsf{ON}]
\]
\end{restatable}
In our proof of \Cref{thm:online-rounding}, we showed that \Cref{alg:online-rounding} achieves value at least $\Omega(V[\mathsf{ON}])$. Consequently,~\Cref{lem:optoff-LP,lem:optonoff} imply the following result.
\begin{theorem}
\label{thm:optoffalg}
\Cref{alg:online-rounding} is an $O\left(\frac{\ln T}{\ln \ln T}\right)$-competitive online algorithm for \AVA under $T$ known i.i.d.~arrivals with each item type arriving an expected constant number of times.
\end{theorem}

\begin{remark}
Under the stronger assumption that $\ex{A_i} = q_i\cdot T = \Omega(\ln (mT) / \eps^2)$ for each of the $m$ item types $i$ (e.g., if $T$ grows while the distribution $\{q_i\}$ remains fixed), the number of arrivals of each item is more concentrated: it is $\ex{A_i}\cdot (1\pm \eps)$ w.h.p. 
Consequently, natural extensions of the arguments above, with a smaller blow-up of the RHS of the constraints in~\eqref{opton-bundle-LP}, imply that \Cref{alg:online-rounding}'s competitive ratio improves to $O(1)$ in this case.
\end{remark}

\section{Hardness Results}\label{sec:hardness}

In this section we provide hardness of approximation results for \AVA
and stark impossibility results for the generalization to \GAVA. 

\subsection{Max-Coverage hardness of \AVA}\label{sec:BBC-hard}

Here we prove that \AVA is as hard as the Max-Coverage problem, even if restricted to the unit-$\rho$ case.

\begin{theorem}[Hardness of \AVA]
  \label{lem:max-cover-hard}
  For any constant $\eps>0$, it is NP-hard to approximate \AVA to a
  factor better than $\big(\frac{e}{e-1}+ \eps\big)$ even for unit-$\rho$ instances.
\end{theorem}

\begin{proof}
  We give a reduction from ``balanced'' instances of the
  \textsc{Max-Coverage} problem. Such an instance consists of a set
  system with $n$ elements and $m$ sets, with each set containing
  $\nicefrac{n}{k}$ elements. A classic result of
  \cite{feige1998threshold} shows that for each $\delta > 0$,
  there exist $n$ and $k \leq n\delta$, such that it is NP-hard to
  distinguish between the following two cases: (a) there exists a perfect partition,
  i.e., $k$ sets in the set system that cover all $n$ elements
  (YES-instances), and (b) no collection of $k$ sets from the set
  system cover more than $n(1-\nf1e+\delta)$ elements (NO-instances).
  We now define a unit-$\rho$ \AVA instance consisting of:
  \begin{enumerate}
  \item $m$ buyers, where each buyer $i_S$ corresponds to a set $S$ in
    the set system,
  \item $k$ identical \emph{choice items}, which have value
    $1+(\eps/2)\cdot \nicefrac{n}{k}$ for every 
    buyer, and
  \item $n$ distinct \emph{element items}, one for each element $e$,
    which has value $1-(\eps/2)$ for the buyers $i_S$ such
    that set $S$ contains element $e$, and value zero for the other buyers.
  \end{enumerate}

  For a YES-instance of \textsc{Max-Coverage}, there is a solution
  with value $k+n$: we can assign both the choice and element items to
  the buyers corresponding to the $k$ sets in the perfect partition,
  thereby getting us value $n+k$. (The excess for each choice
  item can subsidize the deficit for the $\nicefrac{n}k$
  element items assigned to that buyer.) On the other hand, for a
  NO-instance, the $k$ buyers/sets selected by the choice items can
  give value $k$ and also subsidize at most $n(1-\nf1e+\delta)$
  element items with deficit. (No other items with deficit
  can be chosen.) Setting $\delta = \eps/2$ means the NO-instances
  have value at most
  $k+n(1-\nf1e+\delta) + n\eps/2 \leq n(1 - \nf1e + \eps)$. This gives
  a gap between instances with value at least $n$ and at most
  $n(1-\nf1e + \eps)$, proving the theorem.
\end{proof}

\subsection{Clique hardness of \GAVA}

Next, we prove that approximating \GAVA defined in  \eqref{eq:ros} is as hard as approximating the maximum independent set number in a graph. Recall that the objective in \GAVA is to maximize welfare $\sum_{ij} v_{ij} x_{ij}$ subject to the more general return-on-spend (ROS) constraints: 
\begin{gather}
 \forall j,\;\;\;\; \sum_{i} v_{ij} \; x_{ij} \geq \rho_j \cdot \bigg(
  \sum_i c_{ij}\; x_{ij}\bigg).
\end{gather}
Without loss of generality, we scale $c_{ij}$ and ensure that all $\rho_j = 1$. We show the hardness even for the case where costs depend only on the items, i.e., $c_{ij} = c_i$ for each item $i$. (The case where $c_{ij} = c_j$ for each buyer $j$ is much easier---equivalent to the \AVA problem---because we can just fold the $c_j$ term into the $\rho_j$ threshold.)

\begin{theorem}[Hardness of \GAVA]\label{thm:clique-hardness}
  For any constant $\eps>0$, it is NP-hard to approximate \GAVA for
  $n$-buyer instances with $\Omega(n^2)$ 
  items to better than a factor of
  $n^{1-\eps}$.
\end{theorem}

The proof uses a reduction from the Maximum Independent Set problem.
The reduction proceeds as follows: given a graph $G = (V,E)$ with $|V| = n$,
define $M := 2|E|/n^{\eps}$, and construct the
following \GAVA instance.
\begin{enumerate}
  
\item For each vertex $v \in V$, there is a buyer $j_v$ with $\rho_{j_v}=1$.
\item For each vertex $v \in V$, there is a \emph{vertex item} $i_v$ with item cost $c_i := M+\hbox{deg}(v)$, where $\hbox{deg}(v)$ is $v$'s degree in $G$; it has 
  value $M$ for the buyer $j_v$, and zero
  value for all other buyers. 
\item For each edge $e = (u,v) \in E$, there is an \emph{edge item} $i_e$ having zero cost; it has 
  value $1$ for buyers $j_u$ and $j_v$,
  and zero value for all others. 
\end{enumerate}

\begin{proof}[Proof of~\Cref{thm:clique-hardness}]
  If vertex item $i_v$ is allocated to buyer $j_v$, then by the 
  constraints above, all edge items $j_e$ with $e\ni v$ must be allocated to $i_v$.
  Thus, the set of vertices $U\sse V$ whose buyers are sold their respective vertex item is an independent set in $G$.
  Conversely, $U$ can be taken to be any independent set.
  Thus, the maximum value obtained by allocating vertex items is precisely $M\cdot \alpha(G)$.
  On the other hand, any optimal allocation must allocate all edge items, as this does not 
  violate any of the ROS constraints. Combining the above, we have that $OPT=\alpha(G)\cdot M + |E|$, 
  where $\alpha(G)$ is the independence number of $G$, i.e., the size of the maximum independent set of $G$.

  Finally, we use the result that for any constant $\eps>0$, it is
  NP-hard to distinguish between the following two scenarios for an
  $n$-node graph $G$: (a)~$G$ contains a clique on $n^{1-\eps}$ nodes
  (YES instances), and (b)~$G$ contains no clique on $n^{\eps}/2$
  nodes (NO instances)
  \cite{hastad1996clique,zuckerman2006linear}. This means that it is
  NP-hard to distinguish between instances of \GAVA with value at least
  $n^{1-\eps} \cdot M$ (corresponding to YES instances) from those
  with value at most $(n^\eps/2) \cdot M + |E| =  n^\eps \cdot M$
  corresponding to the NO instances, and hence proves the claim.
\end{proof}

The above hardness construction can, with small changes, show the
following hardness results. We defer these additional results' proofs, 
as well as algorithms showing the (near) tightness 
of our lower bounds for general \GAVA, to \Cref{sec:hardness-proofs}.
\begin{restatable}{theorem}{iidhardness}(Hardness of i.i.d.~\GAVA)
  For any constant $\eps>0$, it is NP-hard to $n^{1-\eps}$-approximate
  \GAVA in $n$-buyer instances with $\poly(n)$ items drawn i.i.d. from
  a known distribution.
\end{restatable}

\begin{restatable}{theorem}{bicriteriahardness}(Hardness of Bicriteria \GAVA)\label{bicriteria-hard}
  For any $\eps>0$, it is NP-hard to obtain a solution (which can even be infeasible) to \GAVA that achieves an objective value at least $\tilde{\Omega}(\sqrt{\eps})$ times the optimal
  value (i.e. an $\tilde{O}(1/\sqrt{\eps})$-approximation), 
  while guaranteeing  the cost for each buyer is at most $1+\eps$ times
  their total value, assuming the UGC.\footnote{As usual, the soft-Oh notation hides polylogarithmic factors in its argument: i.e., $\tilde{O}(f) = f\cdot \poly\log(f)$.} 
\end{restatable}

\section{Deferred Proofs of Section \ref{sec:structure}}\label{app:structure}
\subsection{Another (Offline) Reduction to Unambiguous Instances}
\label{sec:unamb-reduction}

In this section we provide an alternative, deterministic method to identify unambiguous sub-instances admitting a high-valued bundling-based solution w.r.t.~the original (entire) instance.

Given any \AVA instance $\calI = (I,J,E)$ where items may be
ambiguous, construct an unambiguous instance $\calI'$ for it by
splitting each ambiguous item $i$ by two copies: the \emph{positive
  copy} $i^+$ that has only the \pedges incident to $i$, and the
\emph{negative copy} $i^-$ that has only the \nedges. Clearly the
optimal value of \AVA on $\calI'$ is at least that on the original
instance $\calI$. 
  
\begin{lemma}
  Any bundle-based solution for the unambiguous instance $\calI'$
  of \AVA can be converted into a solution for instance $\calI$
  having at least half the value.
\end{lemma}

\begin{proof}
  Suppose solution for instance $\calI'$ uses bundles
  $B_1, B_2, \cdots$. Let bundle $B_k$ contain some \pitem
  $i_k$ and some set $S_k$ of \nitems. We create an auxiliary digraph
  whose vertex set corresponds to these bundles. To create the directed edges (arcs),
  consider each item $i \in \calI$: if both the copies of some item
  $i$ from $\calI$ are used in this solution in bundles $B_a, B_b$
  (say the positive copy $i^+$ appears as $i_a$ and the negative copy
  belongs to $S_b$), then add an arc $B_a \to B_b$. By this
  construction, each bundle has a single out-arc, and hence the digraph
  created is a \emph{$1$-tree} (a bunch of components, each having a
  ``root'' which is a single node or a cycle, and then in-trees
  pointing into the vertices of the root). We now show how to remove
  these arcs, losing a factor of $2$ in the value.

  First consider any cycle $C$, and let the arcs correspond to items
  $i_1, i_2, \ldots, i_k$. Just remove the \nitems corresponding to
  these items from the bundles. Each bundle loses one \nitem, whose
  value is at most the value of its \pitem, and hence the value
  corresponding to these items reduces by a factor of at most $2$. The
  remaining arcs form a collection of branchings (directed
  trees). Each such branching has a root bundle, and the bundles fall
  into odd and even levels (with the root at level zero). We can now
  discard either the bundles at odd levels or those at even
  levels, whichever has less value. (The root bundle is an exception: we should only consider the \nitems in this bundle when making the decision.) This solution is feasible for
  $\calI$, because each item in $\calI$ is only used as either a
  \pitem or an \nitem and not both; moreover, we lose at most half
  the value of the items associated with these arcs. 
\end{proof}

\section{Deferred Proofs of Section \ref{sec:rounding-offline}}\label{app:offline}
\largeExcessNotWorthwhile*
\begin{proof}
Fix a bundle $jp$. By Constraint \eqref{cons:RoS}, we have that 
\begin{align*}
    (v_{pj} - \rho_j)\; x_{pjp} \geq \sum_{i\neq p} (\rho_j-v_{ij})\; x_{ijp} \geq \sum_{i\in L^{\beta}_{jp}} (\rho_j-v_{ij})\; x_{ijp} > \sum_{i\in L^{\beta}_{jp}} \beta \; (v_{pj}-\rho_j)\; x_{ijp}.
\end{align*}
Note that if $v_{pj} - \rho_j =0$, then $\sum_i x_{ijp}=0$ by Constraint \eqref{cons:RoS}, and if $v_{pj} - \rho_j >0$ then we can divide the above inequality by $v_{pj}-\rho_j$. Therefore, we have $\frac{1}{\beta} \; x_{pjp} \geq \sum_{i\in L^{\beta}_{jp}} x_{ijp}$. On the other hand, each \nitem $i$ has value at most $v_{ij}\leq \rho_j \leq v_{pj}$, and so 
$$\sum_{i\in L^{\beta}_{jp}} v_{ij}\; x_{ijp}\leq \frac{1}{\beta} \; v_{pj}\; x_{pjp}.$$
The lemma follows by summing both sides over all bundles $jp$.
\end{proof}

\subsection{Adding Side Constraints}

This section is dedicated to the proof of the following theorem.
\multiconstraint*

In what follows, suppose we have $K$ budget constraints, of the form $\sum_{i\to j} \ell_{ij}\leq B^{(\ell)}_j$ for $\ell\in [K]$. When fixing a particular budget constraint, we drop the superscript $\ell$.

First, to capture budget constraints to our bundling LP \eqref{bundle-LP}, we simply introduce the following additional constraints for every resource $\ell$.
\begin{align}
\sum_{i,p} \ell_{ij} \; x_{ijp} & \leq B^{(\ell)}_j \qquad \;\; \qquad \forall \text{ buyer } j, \label{eqn:per-buyer-budget} \\
\sum_{i} \ell_{ij} \; x_{ijp} & \leq B^{(\ell)}_j\cdot x_{pjp}. \qquad \forall \text{ buyer } j, \text{ \pitem \ } p. \label{eqn:per-bundle-budget}
\end{align}
The first constraints simply assert that in expectation, the cost to buyer $j$ is at most their budget, which holds since the same constraint holds for every realization.
The second constraints assert that since the $\ell$-cost of any bundle may not exceed the budget $B^{(\ell)}_j$, the expected cost of a bundle is at most the budget $B^{(\ell)}_j$, times the probability that this bundle is opened, namely $x_{pjp}$.
These constraints are valid for any bundling-based algorithm satisfying both average-value and budget constraints.
We conclude that the LP \eqref{bundle-LP} with the additional constraints \eqref{eqn:per-buyer-budget} and \eqref{eqn:per-bundle-budget} upper bounds the expected value of any average-value and budget-respecting allocation.
On the other hand, the proof of \Cref{lem:bundling} and \Cref{only-p-or-n} imply that the best bundling-based solution (after making the instance unambiguous) is a $4$-approximation of the best solution (of any kind).\footnote{The only delicate point is that budget constraints are downward closed, and since \Cref{lem:bundling} computes a sub-solution of a budget-respecting allocation, this output is itself budget-respecting.} To conclude, we have the following.
\begin{lemma}\label{lem:LP-ROS-plus-budget}
For any \AVA instance $\calI$ with budget constraints, LP \eqref{bundle-LP} together with constraints \eqref{eqn:per-buyer-budget} and \eqref{eqn:per-bundle-budget} applied to the unambiguous instance described in \Cref{only-p-or-n} has value at least $1/4$ of the optimal solution to $\calI$.
\end{lemma}

We now discuss the minor changes to the design and analysis of \Cref{alg:offline-rounding} that allow us to prove a constant approximation with respect to the new LP under the \emph{small-bids assumption}, whereby $\ell_{ij}/B^{(\ell)}_j\leq (\eps \to 0)$, popular in the analysis of \emph{online} BAP (AdWords~\cite{mehta2007adwords}) algorithms.
First, our algorithm computes an optimal solution to \eqref{bundle-LP} with the additional $K$ sets of constraints of \eqref{eqn:per-buyer-budget} and \eqref{eqn:per-bundle-budget} for each of the $K$ budget constraints.
Then, in \Cref{line:permissible-test}, we only add $i$ to the single bundle $jp\in S_i$ if adding $i$ to $jp$ leaves this bundle permissible \emph{and} does not violate any of the $K$ budget constraints.

Now, fix a triple $i,j,p$, and let $\calE_1,\calE_2,\calE_3,\calE_4$ be as in the analysis of \Cref{alg:offline-rounding} (without budgets), and let $\calE^{(\ell)}_5$ be the event that the cost of items allocated in budget $jp$ is no greater than $B^{(\ell)}_j-\ell_{ij}$, i.e., the item $i$ can be added to the bundle $jp$ without violating the $\ell$-th budget constraint of buyer $j$.
We have that $i$ is allocated in bundle $jp$ if all events $\calE_1,\dots,\calE_4$ and $\bigwedge_{\ell} \calE^{(\ell)}_5$ all occur (simultaneously).
The following lower bound on $\pr{\calE^{(\ell)}_5 \mid \calE_1 \land \calE_2 \land \calE_3}$ follows by a similar coupling argument of \Cref{lem:small-excess-has-room} with an imaginary algorithm allocating items multiple times and ignoring constraints, but this time using constraints \eqref{eqn:per-buyer-budget} and \eqref{eqn:per-bundle-budget} in the analysis.

\begin{lemma}\label{lem:budget-has-room}
If $\max_{ij} \ell_{ij}/B^{(\ell)}_j\leq \eps$, then
$\Pr[\calE^{(\ell)}_5 \mid \calE_1 \land \calE_2 \land \calE_3] \geq 1-\frac{2\alpha}{1-\eps}.$
\end{lemma}
\begin{proof}
In what follows, we drop the superscript $(\ell)$, as it is clear from context.
Let $Y_{ijp}$ be the indicator for item $i$ being allocated in bundle $jp$, and let $Y_{ijp} \leq Z_{ijp} = \mathds{1}[S_i\ni \{jp\}]$. Then $Y_{ijp}\leq Z_{ijp}$ realization-by-realization, and moreover $\Pr[Z_{ijp}] = \alpha\cdot x_{ijp}$. Therefore, we immediately have from Constraint \eqref{eqn:per-buyer-budget} and independence of bundles $jp'$ and $jp$ that, recalling that $\calE_1$ is the event that $jp$ is open,
\begin{align*}
\ex{\sum_{i,p'} \ell_{ij}\cdot Z_{ijp'} \;\;\middle|\;\; \calE_1}  = \ex{\sum_{i,p'} \rho_j\cdot Z_{ijp'}  \;\;\middle|\;\;  \calE_1} \leq \alpha\cdot B_j.
\end{align*}
Similarly, by Constraint \eqref{eqn:per-bundle-budget}, we obtain that
\begin{align*}
\ex{\sum_{i} \ell_{ij} \cdot Z_{ijp'} \;\;\middle|\;\; \calE_1 } \leq \frac{\alpha\cdot B_j\cdot x_{pjp}}{x_{pjp}} = \alpha\cdot B_j.
\end{align*}
Consequently, by Markov's inequality, we have that 
\begin{align*}
\pr{\sum_{i,p} \ell_{ij} \; Z_{ijp} \geq (1-\eps)\cdot B_j \,\,\middle|\,\, \calE_1} & \leq \frac{2\alpha\cdot B_j}{(1-\eps)\cdot B_j} \leq \frac{2\alpha}{1-\eps}. 
\end{align*}
On the other hand, if we denote by $\calE'_5$ the event that the imaginary algorithm $\calA'$ that allocates any item $i$ into a bundle $jp\in S_i$ regardless of whether or not $|S_i|=1$ and the allocation remains average-value- and budget-respecting, we have that
$\calE'_5$ and $\calE_1$ are independent of $\calE_2$ and $\calE_3$, and so we have that
\begin{align*}
\pr{\calE_5 \mid \calE_1\land \calE_2 \land \calE_3} \geq \pr{\calE'_5 \mid \calE_1\land \calE_2 \land \calE_3} & = \pr{\calE'_5 \mid \calE_1} \geq 1- \frac{2\alpha}{1-\eps}. \qedhere
\end{align*}
\end{proof}

Generalizing the arguments in \Cref{thm:offline-rounding}, we obtain the following result, implying \Cref{thm:multiconstant}. 

\begin{restatable}{lemma}{ROSandBudget}\label{thm:offline-rounding-budgeted}
\Cref{alg:offline-rounding} with the modifications outlined in this section and with $\alpha=1/3K$ is an $O(K)$-approximation for \AVA and $K$ budget constraints subject to the small bids assumption.
\end{restatable}
\begin{proof}[Proof (Sketch)]
Let $\beta=1/2$. Applying union bound over the $K$ events $\calE^{(\ell)}_5$ and combining \Cref{lem:small-excess-has-room} and \Cref{lem:budget-has-room}, we find that 
$$\pr{\bigwedge_{\ell} \calE^{(\ell)}_5 \land \calE_4 \;\;\middle|\;\; \calE_1 \land \calE_2 \land \calE_3} \geq 1-\frac{\alpha}{1-\beta} - \frac{2K\alpha}{1-\eps} \approx 1-2(K+1)\alpha.$$
The same argument in the proof of \Cref{thm:offline-rounding}, but this time taking $\gamma = \gamma(\alpha,\beta):= \alpha\cdot (1-\alpha)\cdot (1-2(K+1)\cdot \alpha)$ then implies that this modification of \Cref{alg:offline-rounding} outputs a solution of value at least a $\min\{1-\frac{\gamma}{\beta},\;\gamma\} = \min\{1-2\gamma,\;\gamma\}$ fraction of the optimal LP value; i.e., this algorithm is a $1/\min\{1-2\gamma,\;\gamma\}$-approximation.
Taking $\alpha = \frac{1}{3K}$, this yields an $O(K)$ approximation. The bound then follows by \Cref{lem:LP-ROS-plus-budget}.
\end{proof}

\section{Deferred Proofs of Section \ref{sec:OPTon}}\label{app:opton}

Recall that events $\calE_0,\calE_1,\calE_2$ are all independent, and similarly $\calE_1,\calE_2,\calE_3$ are independent (though $\calE_0$ and $\calE_3$ are not independent). 
So, for example, we have the following fact.
\begin{fact}\label{opton-trivial-Si-prob}
$\pr{S_{it^\star}\ni jpt} = \pr{\calE_0\land \calE_1\land \calE_2} = \frac{\alpha\cdot x_{ijp}}{T^2}.$
\end{fact}
\begin{proof}
The first equality follows by definition of the events $\calE_0,\calE_1,\calE_2$.
The second equality follows from independence of these events, as follows.
\begin{align*}\pr{S_{it^\star}\ni jpt} & = \pr{\calE_0}\cdot \pr{\calE_1}\cdot \pr{\calE_2}=q_i\cdot \frac{x_{pjp}}{T}\cdot \frac{\alpha\cdot x_{ijp}}{x_{pjp}\cdot q_i\cdot T} = \frac{\alpha\cdot x_{ijp}}{T^2}. \qedhere
\end{align*}
\end{proof}
\simpleOPTonBound*
\begin{proof}
\Cref{opton-trivial-Si-prob} gives us a closed form for $\pr{\calE_0\land \calE_1 \land \calE_2}$. We now lower bound $\pr{\overline{\calE_3} \;\;\middle|\;\; \calE_0\land \calE_1 \land \calE_2}$.
First, as $\Pr[X>0] = \Pr[X\geq 1] \leq \E[X]$ for any integer random variable~$X\geq 0$ by Markov's inequality,
\[ \Pr[ \overline{\calE_3}] \leq \ex{ \left|S_{it^\star} \setminus \{jpt\} \right|} =
    \sum_{j'} \sum_{p'} \sum_{t'\in [T/2]\setminus\{t\}} \frac{\alpha \cdot  x_{ij'p'}}{T^2} \leq \frac{\alpha\cdot q_i}{2}, \] 
    where the equality follows from \Cref{opton-trivial-Si-prob} (applied to appropriate tuple $(i,t^\star,j',p',t')$), and the last inequality follows from Constraint \eqref{opton-cons:item}.
    On the other hand, since $\calE_0$ and $\calE_3$ are both independent of $\calE_1\land \calE_2$, an application of Bayes' Law tell us that 
$$\pr{\overline{\calE_3} \;\;\middle|\;\; \calE_0\land \calE_1\land \calE_2} = \Pr[ \overline{\calE_3} \mid \calE_0] \leq \frac{\alpha}{2}.$$
    Therefore, we have that 
    $$\pr{\calE_3 \;\;\middle|\;\; \calE_0\land \calE_1\land \calE_2} \geq 1 - \frac{\alpha}{2},$$
    and the lemma follows.
\end{proof}

\optonroom*
\begin{proof}
We consider an imaginary algorithm $\calA'$ that allocates every \nitem $it^\star$ into every bundle $jpt\in S_{it^\star}$ (even when $|S_{it^\star}|>1$ and even if this violates the bundle's average-value constraint of some $jpt\in S_{it^{\star}}$). 
 Coupling $\calA'$ with \Cref{alg:online-rounding} by using the same randomness for both algorithms, we have by \Cref{opton-trivial-Si-prob} that item $i't''$ is allocated to bin $j'p'$ by $\calA'$ with probability precisely 
 $$\pr{S_{it''} \ni j'p'}=q_i\cdot \frac{x_{pjp}}{T}\cdot\frac{\alpha\cdot x_{ijp}}{x_{pjp}\cdot q_i\cdot T} = \frac{\alpha\cdot x_{i'j'p'}}{T^2}.$$
 Now, letting $\calE'_4$ be the event that $\mathsf{BundleAV}_{jpt}$ is satisfied if we were to add $it^\star$ to $jpt$ in Algorithm $\calA'$, we clearly have that $\calE_4\geq \calE'_4$, realization by realization, since $\calA'$ only allocates more items than \Cref{alg:online-rounding}.
 On the other hand, we also have that both $\calE'_4$ and $\calE_1$ are independent of $\calE_3\land \calE_2 \land \calE_0$. Consequently, by Bayes' Law, we obtain the following.
 \begin{align*}
 \Pr[\calE'_4 \mid \calE_0 \land \calE_1 \land \calE_2 \land \calE_3] & = \Pr[\calE'_4 \mid \calE_1].
 \end{align*}
 Now, since the imaginary algorithm $\calA'$ assigns $i't'$ to $jpt$ iff $S_{i't'}\ni jpt$, the set of \nitems allocated to bundle $jpt$ by $\calA'$, denoted by $I'_{jpt}$, satisfies
 \begin{align*}
 \ex{\sum_{i't'\in I'_{jpt}} (\rho_j - v_{i'j}) \,\,\middle\vert\,\, \calE_1 } & = \sum_{i'\neq p}\sum_{t'} (\rho_j - v_{i'j})\cdot \pr{S_{i't'}\ni jpt \mid \calE_1} \\
 & = \sum_{i'\neq p} (\rho_j - v_{i'j})\; \frac{\alpha}{2} \cdot \frac{x_{i'jp}}{x_{pjp}} \\
 & \leq \frac{\alpha}{2} \cdot (v_{pj} - \rho_j).
 \end{align*}
 Above, the second equality used that $\Pr[S_{i't''}\ni j'p't' \mid \calE_1] = \frac{\alpha\cdot x_{i'j'p'}}{x_{p'j'p'}\cdot T^2}$, by \Cref{opton-trivial-Si-prob}. The inequality
 follows from the per-bundle average-value constraint (\Cref{opton-cons:RoS}), together with summation over $t'\in [T/2]\setminus \{t\}$.
 Therefore, by Markov's inequality,
 \begin{align*}
 \Pr[\overline{\calE'_4} \mid \calE_1] = \Pr\left[\sum_{i\in I'_{jp}} (\rho_j - v_{ij}) > (1-\beta)\cdot (v_{pj} - \rho_j) \,\,\middle\vert\,\, \calE_1 \right] \leq \frac{\alpha}{2(1-\beta)}. \label{eqn:markov}
 \end{align*}
 Recalling that $\calE_4\geq \calE'_4$ realization by realization, we conclude with the desired bound, as follows.
 \begin{align*}
 \Pr[\calE_4 \mid \calE_0 \land \calE_1 \land
   \calE_2 \land \calE_3] & \geq \Pr[\calE'_4 \mid \calE_0 \land\calE_1 \land
   \calE_2 \land \calE_3] = \Pr[\calE'_4 \mid \calE_1] \geq 1-\frac{\alpha}{2(1-\beta)}. \qedhere
 \end{align*}
 \end{proof}

\section{Deferred Proofs of Section \ref{sec:OPToff}}\label{app:optoff}
\Aiconcentrated*
\begin{proof}
This is a fairly standard application of Chernoff bound plus union bound, as in the classic balls and bins analysis. Technically, since $A_i\sim \Bin(T,q_i)$ is the sum of independent Bernoullis with $\ex{A}=q_i\cdot T\geq \Gamma$, by the multiplicative Chernoff bound, for $\delta=\kappa-1$ and $\Gamma' := \min(1,\Gamma)$, 
we have that 
\begin{align*}
\pr{A_i\geq (1+\delta)\cdot \ex{A_i}} & \leq \exp\left(-\ex{A_i}\cdot ((1+\delta) \ln(1+\delta)-\delta)\right) \\ 
& = \exp\left(\ex{A_i}\cdot (\kappa-1 - \kappa \log\kappa)\right) \\
& \leq \exp\left(\ex{A_i}\cdot (-\kappa (\ln \kappa - 1))\right) \\
& \leq \exp\left(- \frac{6\ln T}{\ln \ln T} (\ln\kappa-1) \right),
\end{align*}
where the last inequality follows from $\ex{A_i}\geq \Gamma \geq \Gamma'$. Next, since 
\[
\ln\kappa \geq \ln\frac{6}{\Gamma'}+\ln\ln T-\ln\ln\ln T\geq 1-\ln\Gamma'+\frac{\ln\ln T}{2},
\]
where the last inequality relied on $\frac{x}{2}\geq \ln(x)$ for all $x\in \mathbb{R}$, 
we get that
\begin{align*}
\pr{A_i\geq (1+\delta)\cdot \ex{A_i}} & \leq \exp\left(- \frac{6\ln T}{\ln \ln T} \frac{\ln\ln T}{2} +\frac{6\ln T}{\ln \ln T}\ln\Gamma'\right)\\
& \leq \exp\left(-3\ln T +\ln\Gamma'\right) \\
& \leq \frac{\Gamma}{ T^3},
\end{align*}
where the second to last inequality used $\Gamma'\leq 1$ and the last inequality used $\Gamma'\leq \Gamma$.

On the other hand, as $q_i\cdot T\geq \Gamma$ for each $i\in [m]$, we have that $T=\sum_{i\in m} q_i T \geq \Gamma\cdot m$, or put otherwise $m\leq T/\Gamma$. The lemma then follows by union bound.
\end{proof}

\optonoff*
\begin{proof}
Denote by $x$ some solution to \eqref{optoff-bundle-LP}, and note that the RHS of constraints \eqref{opton-cons:item} and \eqref{optoff-cons:nitem} differ by a factor of $2\lceil q_i\cdot T\rceil / q_i\cdot T = O(1) = O(\kappa)$ (using that $\Gamma=\Omega(1)$ and $\kappa = \omega(1)$). Similarly, the RHS of constrains \eqref{opton-cons:bundle-defined-by-p-item} and \eqref{optoff-cons:bundle-defined-by-p-item} differ by a factor of $\lceil q_i\cdot T\cdot \kappa\rceil / q_i\cdot T = O(\kappa)$. In both cases, the RHS in the constraint in \eqref{optoff-bundle-LP} is higher than its counterpart in \eqref{opton-bundle-LP}.
Therefore, the solution $x/O(\kappa)$ (for an appropriate $O(\kappa)$ term) satisfies the aforementioned constraints in \eqref{opton-bundle-LP}, and it is easy to check that it satisfies all other constraints, which are either downward-closed or linear (Constraint \eqref{opton-cons:RoS}). The lemma then follows, since the obtained solution to \eqref{opton-bundle-LP} has value $O(\kappa) = O\left(\frac{\ln T}{\ln \ln T}\right)$ than the original solution to \eqref{optoff-bundle-LP}, $x$.
\end{proof}

\section{Deferred Proofs of Section \ref{sec:hardness}}\label{sec:hardness-proofs}

In this section we provide hardness proofs deferred from
\Cref{sec:hardness}, restated below, together with an algorithm giving
a bicriteria guarantee complementing our bicriteria hardness.

\iidhardness*
\begin{proof}
  We construct an instance similar to that for
  \Cref{thm:clique-hardness}. Given a graph $G$ and parameters $M$ and
  $R \approx \Theta(\eps^2 \ln |E|)$, each vertex item $i_v$ has
  value $M$ and cost $M+ R\cdot \deg_G(v)$, and each edge item $i_e$
  has unit value and zero cost. (We choose $\eps \leq \nf1{(2n^2)}$.)
  The distribution over items is simple: each vertex item appears with
  probability $\frac{1}{2|V|}$, and each edge item with probability
  $\frac{1}{2|E|}$. Now we take $2(1+\nf\eps2)R|E|$ i.i.d.\ samples from
  this distribution.
  \begin{enumerate}
  \item With these many samples, each vertex item is seen at least
    once.
  \item Moreover, we expect to see each edge item $(1+\eps/2)R$ times, and
    concentration implies that each edge is seen at least $R$ times
    and at most $(1+\eps)R$ times (with high probability).
  \end{enumerate}
  We claim that if we allocate any vertex item $i_v$ to $j_v$, the ROS constraint
  for buyer $j_v$ (for vertex $v$ having degree $deg_G(v) = d$
  in the graph $G$, say) requires us to pick at least $dR$ edge items
  incident to $v$.  Every edge contributes at most $R(1+\eps)$ edge
  items, so even if we get the maximum number of items from all but
  one edge, that last edge needs to contribute at least
  $dR-(d-1)R(1+\eps) \geq R(1- d\eps) \geq R(\nf12 + \eps)$ items. But
  then this ``underpaying'' edge can only contribute $R/2$ to its
  other endpoint, which is not enough to satisfy that vertex's
  deficit. This enforces the independent set condition.  Finally, the
  value we achieve lies between $\alpha M$ and
  $\alpha M + (1+\eps)R|E|$; setting $M$ to be large enough gives the
  claimed gap between YES and NO instances with high probability.
\end{proof}


\bicriteriahardness*
\begin{proof}
  We construct the same instance as for \Cref{thm:clique-hardness}:
  each vertex item $i_v$ having value $M$ and cost $M+\deg_G(v)$
  potentially subsidized by all the edge items $i_e$ of unit value and zero
  cost around it. Consider a $d$-regular graph $G$, set $M := d^2$ and
  $\eps := \nicefrac{1}{(M+d+1)} = \Theta(1/d^2)$. Then any selected
  vertex item $i_v$ must pick all its incident edge items, else the
  value would be at most $M+d-1 < (1-\eps)(M+d)$, violating the ROS constraint
  by more than a factor of $(1-\eps)$. An argument identical to
  \Cref{thm:clique-hardness} shows that if the maximum independent
  set has size $\alpha$ (which is at least $n/(d+1)$ by Turan's
  theorem), then the achieved value in the \GAVA instance is at least
  $\alpha M$ and at most $\alpha M + nd/2 \leq 2\alpha M$. Finally,
  we use the result of \cite{khot2001improved} that approximating the
  the Independent Set problem in $d$-regular graphs to better than a
  factor of $\tilde{O}(d)$ would violate the Unique Games Conjecture
  to infer our $\tilde{\Omega}(\sqrt{\eps})$ lower bound. 
\end{proof}

Observe that a polynomial dependence on $1/\eps$ in the approximation
ratio is straight-forward.

\begin{lemma}
For any $\eps>0$, there exists a linear-time algorithm which computes
a (nearly feasible) solution whose objective value is at least $\eps$ 
times the optimal value of any \GAVA instance while guaranteeing that the cost for each buyer is at most $1+O(\eps)$ times their total value.
\end{lemma}
\begin{proof}
  We assign every item $j$ to
  $\arg\max\{v_{ij} \mid v_{ij}\geq c_{ij}(1-\eps)\}$, if this set is
  non-empty, and leave $j$ unallocated otherwise.  Let $J_i$ be the
  set of items $j$ allocated to buyer $i$ in an optimal
  (ROS-constraint respecting) assignment, and let $L_i\subseteq J_i$ be the set
  of low-ROS items for $i$ in this assignment, i.e., those satisfying
  $v_{ij}\leq c_{ij}(1-\eps)$.  Then,
  $$\sum_{j\in L_i} v_{ij} \leq \sum_{j\in L_i} c_{ij}(1-\eps) \leq
  \sum_{j\in J_i} c_{ij}(1-\eps) \leq \sum_{j\in J_i}
  v_{ij}(1-\eps).$$ We conclude that
  $\sum_{j\in L_i\setminus J_i} v_{ij} \geq \eps\cdot \sum_{j\in J_i}
  v_{ij}$ for each buyer $i$.  The above greedy solution allocates
  items $j$ in $\cup_i J_i$ to buyers who value $j$ at least as much
  as the buyer that $j$ is sold to in the optimal assignment, and so
  the overall objective value is at least
  $\sum_i \eps\cdot \sum_{j\in J_i}v_{ij}$, i.e., this is a
  $1/\eps$-approximation.  That this solution
  $(1+O(\eps))$-approximately satisfies ROS constraints is obvious,
  since it does so on a per-item/buyer pair basis.
\end{proof}

We conclude with a brief observation, whereby our $n^{1-\eps}$ approximation lower bounds are essentially tight.
Indeed, an $O(n)$ approximation for \GAVA is nearly trivial: Pick the
(approximately) highest-value allocation to a single buyer, by
allocating it all of its \pedges, and then allocating the
value-maximizing \nedges by running any constant-approximate knapsack
algorithm, e.g., the basic $2$-approximate algorithm
\cite{vazirani2001approximation}, giving a $2n$-approximation.

\end{document}